\documentclass[10pt,a4paper]{article}
\usepackage{amssymb}
\usepackage{amsmath}
\usepackage{amsfonts}
\usepackage{amsthm}
\usepackage{latexsym}
\usepackage{mathrsfs}
\usepackage{stmaryrd}
\usepackage[dvips]{epsfig}
\usepackage{setspace}
\usepackage{float}
\usepackage{bm}
\usepackage{tikz}
\usepackage{enumerate}
\usepackage{subfigure}
\usepackage{nomencl}
\usepackage{authblk}

\usepackage{textcomp}
\usepackage{scrextend}
\usepackage[toc,page]{appendix}
\usepackage{comment}
\usepackage{xcolor}
\usepackage[normalem]{ulem}
\newcommand\omicron{o}

\theoremstyle{plain}
\newtheorem{proposition}{Proposition}
\newtheorem{lemma}{Lemma}
\newtheorem{theorem}{Theorem}

\newtheorem{corollary}{Corollary}
\newtheorem*{main}{Theorem}
\newtheorem*{definition}{Definition}
\newtheorem{remark}{Remark}

\def\bmg{{\bm g}}
\def\bmh{{\bm h}}

\def\bmo{{\bm o}}

\def\bmB{{\bm B}}

\def\bmE{{\bm E}}
\def\bmF{{\bm F}}

\def\bmH{{\bm H}}
\def\bmK{{\bm K}}

\def\bmQ{{\bm Q}}

\def\bmS{{\bm S}}

\newcounter{mnotecount}

\newcommand{\mnotex}[1]
{\protect{\stepcounter{mnotecount}}$^{\mbox{\footnotesize $\bullet$\themnotecount}}$ 
\marginpar{
\raggedright\tiny\em
$\!\!\!\!\!\!\,\bullet$\themnotecount: #1} }

\newcommand{\notimplies}{%
  \mathrel{{\ooalign{\hidewidth$\not\phantom{=}$\hidewidth\cr$\implies$}}}}

\allowdisplaybreaks

\begin{document}

\title{\textbf{The conformal Killing spinor initial data equations}}

\author[1]{E. Gasper\'in \footnote{E-mail address:{\tt edgar.gasperin@tecnico.ulisboa.pt}}}
\author[2]{J. L. Williams \footnote{E-mail address:{\tt jrrodwilliams@gmail.com}}}
\affil[1]{CENTRA, Departamento de F\'isica,
  Instituto Superior T\'ecnico IST, Universidade de Lisboa UL, Avenida
  Rovisco Pais 1, 1049 Lisboa, Portugal.}
\affil[2]{Department of Mathematical Sciences, University of Bath, Claverton Down, Bath BA2 7AY, United Kingdom.}

\maketitle
\begin{abstract}
  We obtain necessary and sufficient conditions for an initial data
  set for the \emph{vacuum conformal Einstein field equations} to give
  rise to a spacetime development in possession of a Killing spinor.  The
  fact that the conformal Einstein field equations are used in our
  derivation allows for the possibility of the initial hypersurface
  $\mathcal{S}$ intersecting non-trivially with (or even being a subset of)
  null infinity $\mathscr{I}$.  For conciseness, these
  conditions are derived assuming that the initial hypersurface is
  spacelike. Hence, in particular, these \emph{conformal Killing spinor initial data} equations encode necessary and
  sufficient conditions for the existence of a Killing spinor in the
  development of asymptotic initial data on spacelike components of
  $\mathscr{I}$.
\end{abstract}

\section{Introduction}
\label{sec:Introduction}

The discussion of symmetries  in General
Relativity is ubiquitous. From the question of the integrability of the geodesic
equations to the existence of explicit solutions to the Einstein field
equations and the black hole uniqueness problem, symmetries play an important role.   
Symmetry assumptions are usually incorporated into
 the Einstein field equations ---which in vacuum read
\begin{equation}
\tilde{R}_{ab}=\lambda \tilde{g}_{ab}
\label{EFEVacuum}
\end{equation} 
---through the use of Killing vectors.  From the spacetime point of
view, the existence of Killing vectors allows one to perform
\emph{symmetry reductions} of the Einstein field equations ---see \cite{Wei90a} for
instance. This approach has been exploited in classical
uniqueness results such as \cite{Rob75b}.  Closely related to the
black hole uniqueness problem, characterisations and classifications
of solutions to the Einstein field equations usually exploit the
symmetries of the spacetime in one way or another, e.g. in the
characterisations of the Kerr spacetime via the \emph{Mars--Simon
  tensor} ---see \cite{Mar99,Mar00,Sim84, MarPaeSenSim16, MarPeo22}.
On the other hand, from the point of view of the Cauchy problem,
symmetry assumptions should be imposed only at the level of initial
data. In this regard, symmetry assumptions can be phrased in terms of
the \emph{Killing vector
initial data}.  The Killing vector initial data (KID) equations are a
system of PDEs, defined over a spacelike hypersurface $\tilde{\mathcal{S}}$ and with coefficients computable in terms of the first and second fundamental forms $\tilde{\bmh}$ and $\tilde{\bmK}$,
whose solutions (whenever they exist) correspond to initial data for Killing vectors on the ensuing spacetime development $(\tilde{\mathcal{M}},\tilde{\bmg})$, in the form of lapse-shift pairs ---see \cite{BeiChr97b}.
While Killing vectors
play a central role in the discussion of the symmetries, their existence is sometimes not enough to encode all the symmetries and conserved quantities enjoyed by a spacetime e.g. the Carter constant in the Kerr spacetime. One approach to unraveling some of
these \emph{hidden symmetries} is to consider a more fundamental type of object, namely \emph{Killing spinors}, denoted
here by $\tilde{\kappa}_{AB}$. For vacuum spacetimes, the existence of
a Killing spinor directly implies the existence of a Killing
vector. The \emph{Killing spinor initial data equations} have been
derived in the \emph{physical framework} ---i.e. where the manifold on interest is a solution to the Einstein
field equations--- in \cite{GarVal08c}. These equations have been
used in the construction of a geometric invariant which detects
whether or not an initial data set corresponds to initial data for the
Kerr spacetime ---see \cite{BaeVal10a,BaeVal10b,BaeVal10c,BaeVal11b}.
This analysis has also been extended to include suitable classes of
matter ---see \cite{ValCol16} for an analogous characterisation of
initial data for the Kerr-Newman spacetime.  In these
characterisations, some asymptotic conditions on the initial data are
required. These conditions usually take the form of decay assumptions
on $\tilde{\bmh}$, $\tilde{\bmK}$ and $\tilde{\bm\kappa}$ on
$\tilde{\mathcal{S}}$, given in terms of asymptotically Cartesian
coordinates.  Alternatively, following \emph{Penrose's proposal}, the
asymptotic region of the spacetime is to be studied in a geometric way
through conformal compactifications.  In this approach one starts with
a \emph{physical spacetime} $(\tilde{\mathcal{M}},\tilde{\bmg})$ where
$\tilde{\mathcal{M}}$ is a 4-dimensional manifold and $\tilde{\bmg}$
is a Lorentzian metric which is a solution to the Einstein field
equations.  Then, one introduces an \emph{unphysical spacetime}
$(\mathcal{M},\bmg)$ into which $(\tilde{\mathcal{M}},\tilde{\bmg})$
is conformally embedded:
$\varphi: \tilde{\mathcal{M}} \rightarrow \mathcal{M}$ such that
\begin{equation} \label{eqn:Chapter:Introduction:ConformalRescaling}
\varphi^{*}\bmg=\Xi^2\tilde{\bmg}.
\end{equation}
where $\Xi$ is the \emph{conformal factor}. For so-called \emph{asymptotically simple} spacetimes, one can choose $\Xi$ such that the metric
 $\bmg$ is well defined at the points where $\Xi=0$, these points being at infinity from the physical spacetime perspective.
 \noindent The set of points where the conformal factor vanishes
 is called the conformal boundary and the hypersurface defined by
\[
 \mathscr{I} := \big\{p \in \mathcal{M} \hspace{0.2cm}| \hspace{0.2cm}
 \Xi(p)=0 , \hspace{0.2cm} \mathbf{d}\Xi(p) \neq0\big\}
\]
is called null infinity. In vacuum, the causal nature of this
hypersurface is determined by the sign of the cosmological constant,
being null, spacelike or timelike if $\lambda$ is zero, positive or
negative respectively. The notion of null infinity comprises a powerful tool
for the analysis of asymptotic properties of spacetimes, with potential applications to various open problems in General Relativity.
In making use of these ideas, however, one has to contend with the fact that the Einstein field
equations are not conformally invariant. Moreover, a direct
computation using the conformal transformation formula for the Ricci
tensor shows that the vacuum Einstein field equations
\eqref{EFEVacuum} lead to an equation which, since it includes $\Xi^{-1}$-terms, is formally singular at
the conformal boundary. An approach to remedying this problem was
given in \cite{Fri81a} where a regular set of equations for the
unphysical metric $\bmg$ was derived. These equations are known as the
\emph{conformal Einstein field equations} (CFEs).  The crucial
property of these equations is that they are regular even at the points
where $\Xi=0$, and a solution thereof implies a
solution to the Einstein field equations wherever $\Xi\neq 0$ ---see \cite{Fri81a,Fri83}
and \cite{CFEbook} for a comprehensive discussion. 
The CFEs have found application in
the stability analysis of spacetimes ---see for instance \cite{Fri86b,
  Fri86c} for the proof of the global and semi-global non-linear
stability of the de Sitter and Minkowski spacetimes, respectively.
From the point of view of this article, the main advantage of the conformal
(unphysical) approach to the Einstein field equations is that null
infinity $\mathscr{I}$, being a submanifold of $(\mathcal{M},\bmg)$, is a bonafide hypersurface on which to prescribe data, to be evolved using
(regular) evolution equations. This set
up is particularly attractive in cases where $\lambda>0$, in
which, given the appropriate conditions (sufficient decay of matter fields
at infinity), null infinity is a spacelike hypersurface, allowing for one to pose an \emph{asymptotic initial value problem}: an
initial value problem where the initial hypersurface is $\mathscr{I}$
---see \cite{GasVal17,MarPaeSenSim16, LueVal09}.

\medskip

In \cite{GarKha19}, the authors generalise the KID equations to the broader class of conformal Killing vectors
and in \cite{Gar22} these equations are used to characterise initial data for PP-wave spacetimes. It should be noted that, as the analysis is carried out in the \emph{physical framework}, these conditions only apply to solutions of the Einstein field equations; in particular, the initial hypersurface does not extend to $\mathscr{I}$. On the other hand, the conformally-regular counterpart of the KID equations was derived in \cite{Pae14a}; see also \cite{MarPeo22}
for a generalisation to higher spacetime dimensions.
That is to say, intrinsic conditions on an initial hypersurface
$\mathcal{S}\subset \mathcal{M}$ of the unphysical spacetime $(\mathcal{M},\bmg)$ ---a solution of the CFEs--- are found such that the development of the data gives rise to a conformal Killing vector on $(\mathcal{M},\bmg)$, which moreover corresponds to a Killing vector of the physical spacetime $(\tilde{\mathcal{M}},\tilde{\bmg})$. This construction, in contrast to the previous work, allows for the possibility of $\mathcal{S}$ intersecting non-trivially with, or even being a subset of, 
$\mathscr{I}$.  

\medskip

As previously mentioned, in the case of Petrov type
D spacetimes such as the Kerr-de Sitter spacetime, the symmetries of
the spacetime are closely related to the existence of Killing spinors.
Hence, a natural question is whether there exists a conformal
counterpart ---i.e. in the unphysical framework--- of the Killing spinor
initial data equations introduced in \cite{GarVal08c}. In other words, what are the extra conditions that one has to impose on an initial data set for the CFEs so that the arising development contains a Killing spinor? This question is answered in
this article by deriving such conditions, these being termed the
\emph{conformal Killing spinor initial data equations}. 

\medskip
The main result of this article, the more precise statement of which
can be found in Theorem \ref{Theorem_KS}, is summarised informally
in the following:

\begin{main}\label{TheoremSummary}

If the conformal Killing spinor initial data equations
\eqref{CSKIDs} admit a solution on an open set $\mathcal{U}\subset
\mathcal{S}$, where $\mathcal{S}$ is a spacelike hypersurface on which
initial data for the conformal Einstein field equations has been
prescribed, then there exists a Killing spinor on some open (spacetime) neighbourhood $\mathcal{W}$ of $\mathcal{U}$ contained in the
domain of dependence, $\mathcal{W}\subseteq\mathcal{D}^+(\mathcal{U})$.
\end{main}

   The core of the proof of this theorem is
   obtaining a closed system of homogeneous wave equations for certain fields
   (so called \emph{zero-quantities}) encoding the existence of a Killing spinor. Although these wave
   equations hold regardless of the causal character of $\mathcal{S}$,
   when obtaining conditions intrinsic to $\mathcal{S}$ we assume for conciseness
   that it is spacelike.  A similar
   computation could be performed on a hypersurface $\mathcal{S}$ with
   a different causal character, allowing for applications to
   the black hole uniqueness problem in general.
   In the present set up, of a spacelike $\mathcal{S}$,
   the equations derived here have applicability in the asymptotic characterisation
   of the Kerr-de Sitter spacetime ---which would comprise a spinorial analogue of \cite{MarPaeSenSim16}---
   in terms of the existence of a Killing spinor at $\mathscr{I}$. 
\medskip 

    Although the main objective of the present paper is the valence-2 Killing spinor case, the analogous conditions encoding the
    existence of a valence-1 Killing spinor ---the \emph{conformal
  twistor initial data equations}--- are also derived.  The
    latter serves as a warm-up exercise for the valence-2 case in which one can already see some of the essential features of the analysis.

\subsection*{Overview of the article}
  Section \ref{Background} summarises relevant background material:
  subsection \ref{NotationAndSpinorFormalism} fixes the conventions
  and notation and gives an abridged discussion of the main spinorial
  identities to be used and the space spinor formalism; subsection
  \ref{Sec:KillingSpinors} gives an overview of Killing spinors and
  their conformal properties; subsection \ref{Sec:CFEs} introduces the conformal
  Einstein field equations (CFEs).  In section
  \ref{conformalTwistorKID} the conformal
  twistor (i.e. valence-1 Killing spinor) initial data equations are derived. In section \ref{conformalKSKID} the conformal (valence-2) Killing initial data
  equations are derived and discussed.
  
  \medskip 
  
  Many of the more involved computations in this article were
facilitated through the {\tt xAct} suite in {\tt Mathematica}.

\subsection*{Notations and conventions}

Throughout this article, $(\mathcal{M}, \bmg)$ will denote a
4-dimensional manifold equipped with a Lorentzian metric $\bmg$ of
signature $(+, -, -, -)$, with associated Levi-Civita connection
$\nabla$.  Moreover, $(\mathcal{M}, \bmg)$ is assumed to be
globally-hyperbolic. The Upper case Latin indices ~$_{ABC\cdots
  A'B'C'}$~ will be used as abstract indices of the \emph{spacetime
spinor} algebra and $\epsilon_{AB}$ will denote the skew-symmetric spinor metric. The bold numerals ~$_{\bm0\bm1\bm2\cdots}$~
denote components with respect to a fixed spin dyad $ o^A:=
\epsilon_{\bm0}{}^A,\iota^A:=\epsilon_{\bm1}{}^A $ ---see Penrose \&
Rindler \cite{PenRin84} for further details.  Lower case Latin indices
$_{a,b,c...}$ will be used as abstract tensor indices. Our curvature conventions are fixed by
\[\nabla_{a}\nabla_{b}\kappa^c-\nabla_{b}\nabla_{a}\kappa^c=R{}^{c}{}_{dab}\kappa^{d}.\]
The future domain of dependence of an achronal set $\mathcal{A}$ will
be denoted by $\mathcal{D}^{+}(\mathcal{A})$.

\section{Background}
\label{Background}

In this section, we give an recap of spacetime and space spinor
calculus, in addition to giving a brief introduction to Killing
spinors and the conformal Einstein field equations.

\subsection{Spinorial formalism in a nutshell}
\label{NotationAndSpinorFormalism}

Since $(\mathcal{M}, \bmg)$ is, by assumption, globally-hyperbolic and of signature $(+, -, -, -)$, it
admits a spinor structure ---see Proposition $4$ in
\cite{CFEbook}. 

\medskip

For spinors, the curvature conventions are fixed via the spinorial
Ricci identities which will be written in accordance with the above
convention for tensors.  Recall that the commutator of
covariant derivatives $[ \nabla_{AA'},\nabla_{BB'}]$ can be expressed
in terms of the symmetric operator $\square_{AB}$ as
\[
[ \nabla_{AA'},\nabla_{BB'}]= \epsilon_{AB}\square_{A'B'} +
\epsilon_{A'B'}\square_{AB},
\]
where
\[
\square_{AB} := \nabla_{Q'(A} \nabla_{B)}{}^{Q'}.
\]
 The action of the symmetric operator $\square_{AB}$ on valence-1
 spinors is encoded in the spinorial Ricci identities
\begin{subequations}
\begin{eqnarray}
&& \square_{AB}\xi_{C}=-\Psi_{ABCD} \xi^{D} +
  2\Lambda\xi_{(A}\epsilon_{B)C},
 \label{SpinorialRicciIdentities1} \\
&& \square_{A'B'}\xi_{C}=-\Phi_{CA A' B'}\xi^{A},
\label{SpinorialRicciIdentities2}
\end{eqnarray}
\end{subequations}
where $\Psi_{ABCD}$, $\Phi_{AA'BB'}$ and $\Lambda$ are the standard curvature spinors of the standard Newmann--Penrose (NP) formalism, namely the Weyl spinor, tracefree Ricci spinor and the scalar curvature\footnote{More precisely, $\Lambda=R/24$, with $R$ the Ricci scalar curvature.}, respectively.  The above identities can be extended to higher valence
spinors in the obvious way; further discussion (albeit using
slightly different conventions) can be found in \cite{Ste91}. A
related identity which will be used in the following
discussion is
\begin{equation}\label{DecomposeDoubleDerivativeContracted}
\nabla_{AQ'}\nabla_{B}{}^{Q'}=\square_{AB}+
\tfrac{1}{2}\epsilon_{AB}\square,
\end{equation}
where $\square_{AB}$ is the symmetric operator defined above and
$\square := \nabla_{AA'}\nabla^{AA'}.$
\medskip 

To keep the discussion self-contained, we briefly recall the space spinor
formalism, originally introduced in \cite{Som80}; see also \cite{GarVal08c,BaeVal10b,CFEbook}.  Let $\tau^{AA'}$
denote the spinorial counterpart of a timelike vector $\tau^{a}$,
normal to a spacelike hypersurface $\mathcal{S}$ and normalised so
that $\tau_{a}\tau^{a}=2$.  Then, it follows that
$\tau_{AA'}\tau^{AA'}=2$ and, consequently,
\[\tau_{AA'}\tau_B{}^{A'}=\epsilon_{AB}.\]
Given a spacetime spinor $u_{AA'}$, its space spinor decomposition
reads
\[
u_{AA'}= \tfrac{1}{2}\tau_{AA'}u-\tau^{B}{}_{A'}u_{AB},
\]
where $u:=\tau^{AA'}u_{AA'}$ and $u_{AB}:=\tau_{(A}{}^{B'}u_{B)B'}$.
This split extends to higher valence spinors in an analogous way
---see \cite{GarVal08c,BaeVal10b,CFEbook}. Similarly, the covariant
derivative $\nabla_{AA'}$ is then decomposed into the \emph{normal}
and \emph{Sen} derivatives:
\begin{align*}
  \nabla_{\bm\tau} := \tau^{AA'}\nabla_{AA'},\qquad \mathcal{D}_{AB}:=
  \tau_{(A}{}^{A'}\nabla_{B)A'}.
\end{align*}
Though we will not need them here, for completeness we note that the
\emph{Weingarten} spinor and the \emph{acceleration} of the congruence
are then defined by
\[K_{ABCD} := \tau_{D}{}^{C'} \mathcal{D}_{AB}\tau_{CC'},\qquad K_{AB} := \tau_{B}{}^{C'} \nabla_{\bm\tau}\tau_{AC'}.
\]
The distribution induced by $\tau_{AA'}$ is integrable if and only
$K^D{}_{(AB)D}=0$, in which case $K_{ABCD}$ describes the extrinsic
curvature of the resulting foliation.
The Sen connection is related to the
intrinsic Levi-Civita connection, $D$, as follows
\[\mathcal{D}_{AB}\eta_C = D_{AB}\eta_C + \tfrac{1}{2}K_{ABC}{}^D\eta_D. \]

\subsection{Killing spinors}\label{Sec:KillingSpinors}

Let $(\tilde{\mathcal{M}},\tilde{\bmg})$ also be a
4-dimensional manifold equipped with a Lorentzian metric
$\tilde{\bmg}$ and denote by $\tilde{\nabla}$ its associated
Levi-Civita connection. Later, we will reserve the $~\tilde{\cdot}~$
notation for a vacuum spacetime ---that is to say, a solution of the
vacuum Einstein field equations \eqref{EFEVacuum}. For much of the present section, however, no such restriction is necessary.
\medskip

A totally symmetric
$\tilde{\kappa}_{A_1...A_p}=\tilde{\kappa}_{(A_1...A_p)}$ valence$-p$
spinor is said to be a (valence$-p$) \emph{Killing spinor} if is
satisfies the following equation
\begin{equation}\label{qValenceKillingspinor}
\tilde{\nabla}_{Q'(Q}\tilde{\kappa}_{A_1...A_p)}=0.
\end{equation}
An important property of the Killing spinor equation is that it is
conformally invariant, in other words if $\bmg$ is conformally related
to $\tilde{\bmg}$ ---namely $\bmg=\Xi^2\tilde{\bmg}$--- then
${\kappa}_{A_1...A_p}=\Xi^2 \tilde{\kappa}_{A_1...A_p}$ satisfies
\[{\nabla}_{Q'(Q}{\kappa}_{A_1...A_p)}=0.\]
\medskip
\noindent In this paper we will only focus only the cases $p=1$ and
$p=2$.  If $p=1$, the equation
\begin{equation}\label{TwistorEq}
  \tilde{\nabla}_{Q'(Q}\tilde{\kappa}_{A)}=0.
\end{equation}
is usually referred as the \emph{twistor equation}, and a solution
referred to as a \emph{twistor}; we will follow this naming convention
here.  The valence-2 case, on the other hand, will be referred to
simply as the \emph{Killing spinor case}.  Namely, we will say that a
symmetric valence$-2$ spinor,
$\tilde{\kappa}_{AB}=\tilde{\kappa}_{(AB)}$, is a \textit{Killing
  spinor} if it satisfies the equation
\begin{equation}
\tilde{\nabla}_{A'(A}\tilde{\kappa}_{BC)}=0.
\end{equation}
The Killing spinor equation and twistor equations are, in general,
overdetermined; in particular, they imply the so-called
\textit{Buchdahl constraints}.  In the twistor case ($p=1$), the Buchdahl constraint takes
the form
\[
\tilde{\kappa}^D\Psi_{ABCD}=0,
\]
while in the Killing spinor case ($p=2$) it takes the form
\[
\tilde{\kappa}_{(A}{}^Q\Psi_{BCD)Q}=0,
\]
where $\Psi_{ABCD}$ denotes the Weyl spinor, which is conformally invariant.
This constraint restricts $\Psi_{ABCD}$ to be algebraically special, in particular of Petrov type D, N or O. In the twistor case, which can be considered a degenerate case in which $\kappa_{AB}\kappa^{AB}=0$ (implying that $\kappa_{AB}=\kappa_A\kappa_B$ for some twistor $\kappa_A$), the spacetime is necessarily of Petrov type N or
O, hence restricting its utlity in black
hole characterisation. On the other hand, given a vacuum spacetime $(\tilde{\mathcal{M}},\tilde{\bmg})$ of Petrov type D,
there is an explicit
formula for a Killing spinor: choosing an adapted dyad $\lbrace \bmo,
\bm\iota\rbrace$ for which $\psi_{ABCD}=\psi
o_{(A}o_{B}\iota_C\iota_{D)}$, the following expression
\[\tilde{\kappa}_{AB} = \psi^{-1/3}o_{(A}\iota_{B)}\]
yields a Killing spinor. Indeed, the fact that the Killing spinor equation is
satisfied follows from the vacuum Bianchi identity
$\tilde{\nabla}^A{}_{A'}\Psi_{ABCD}=0$ by a short calculation ---see
\cite{PenRin84, WalkerPenrose70} for more details.
\medskip

Although the Killing spinor equation is conformally invariant, note that one cannot simply transcribe the analysis of \cite{GarVal08c,
  BaeVal10b} into the conformal setting since we would like to allow for the possibility that $\mathcal{S}\cap\mathscr{I}\neq \emptyset$; points in $\mathscr{I}$ do not have corresponding points in the physical spacetime.
Moreover, we shall see that the method given here departs substantially from that of \cite{GarVal08c,
  BaeVal10b}, owing in part to the fact that the Einstein field
equations are not conformally invariant, an important difference being that the set of
\emph{zero-quantities} used to encode the existence of a Killing spinor are different. The results of
\cite{GarVal08c} can however be recovered from the analysis presented here by
setting $\Xi = 1$. The need for a different
set of Killing spinor zero-quantities in the
conformal case can be traced back to the observation that in
$(\mathcal{M},\bmg)$ the vector
$\xi_{AA'}=\nabla_{A'}{}^{Q}\kappa_{QA}$ does not correspond to a Killing (or even a \emph{conformal} Killing) vector. Although for a general Lorentzian manifold
this vector appears not to have any clear geometric significance, a
by-product of the present analysis is that, for conformally Einstein
manifolds (i.e. solutions to the CFEs)
the vector $\xi_{AA'}$ represents a \emph{collineation} of the rescaled Weyl curvature ---see \cite{KatLevDav69} for definitions of curvature
collineations.  Once the existence of a Killing spinor is established
one can use the conformal factor $\Xi$, the Killing spinor
$\kappa_{AB}$ and $\xi_{AA'}$ to construct a conformal Killing vector
$X_{a}$ associated to a Killing vector $\tilde{X}_{a}$ of the physical
spacetime $(\tilde{\mathcal{M}},\tilde{\bmg})$ ---see Remark \ref{GeometricSignificance}, later. In the analysis of
\cite{GarVal08c}, the fact that
$\tilde{\xi}_{AA'}=\tilde{\nabla}_{A'}{}^{Q}\tilde{\kappa}_{QA}$ is a
Killing vector is crucial;
indeed, it motivates the introduction of
$\tilde{S}_{ab} := \tilde{\nabla}_{(a}\tilde{\xi}_{b)}$ as a
zero-quantity.  Similarly,
in the work of \cite{ValCol16}, where the results of \cite{GarVal08c}
are generalised to the case where $(\tilde{\mathcal{M}},\tilde{\bmg})$
satisfies the Einstein-Maxwell equations, the condition
$\tilde{S}_{ab}=0$ is satisfied by virtue of an assumed \emph{matter
alignment condition}. In the conformal setting of interest in this article, the analogous quantity $S_{ab}$ is not as geometrically
motivated as in the physical cases and its usage as a variable in the
system does not lead to a closed system of explicitly regular
homogeneous wave equations. Again, the adjective ``regular" refers to the absence of formally singular terms, such as $\Xi^{-1}$, in the
equations. Instead, the quantity that is central for the present
analysis turns out to be the so-called \emph{Buchdahl zero-quantity} (and
derivatives thereof), the vanishing of which relates the existence of Killing spinors with the Petrov type of $(\mathcal{M},\bmg)$, in line with the above discussion.

\begin{remark}
  \emph{ The notion of Killing spinors is related to that of
  Killing--Yano tensors. Given a Killing spinor $\tilde{\kappa}_{AB}$, if the quantity $\tilde{\xi}_{AA'}$ is
  Hermitian
  then
  one can construct the spinorial counterpart of a \emph{Killing--Yano
  tensor} $\tilde{\Upsilon}_{ab}$ ---i.e. an antisymmetric $2-$tensor
  satisfying $\tilde{\nabla}_{(a}\tilde{\Upsilon}_{b)c}=0$--- as
  follows
\[\tilde{\Upsilon}_{AA'BB'}=i(\tilde{\kappa}_{AB}\bar{\tilde{\epsilon}}_{A'B'}
-\bar{\tilde{\kappa}}_{A'B'}\tilde{\epsilon}_{AB}).\] 
Conversely,
given a Killing--Yano tensor, one can construct a Killing spinor
---see \cite{ValCol16,McLBer93,PenRin86}.  The (Killing--Yano) tensor
counterpart of the Killing spinor initial data result of \cite{GarVal08c}
has been recently derived in \cite{GarKha19a}.  }
\end{remark}
From now on, $(\tilde{\mathcal{M}},\tilde{\bmg})$ will be reserved for the \emph{physical spacetime},
while $(\mathcal{M},\bmg)$ will refer to the \emph{unphysical spacetime}, related to
$(\tilde{\mathcal{M}},\tilde{\bmg})$ via $\bmg=\Xi^2\tilde{\bmg}$ ---as is customary, in a slight abuse of notation, the pullback $\varphi^*$ of the
 embedding $\varphi:
\tilde{\mathcal{M}}\rightarrow\mathcal{M}$ will be omitted.

\subsection{The conformal Einstein field equations}
\label{Sec:CFEs}

The conformal Einstein field equations (CFEs), first given in \cite{Fri81a}, are a conformal
reformulation of the Einstein field equations. In other words, given a
spacetime $(\tilde{\mathcal{M}},\tilde{\bmg})$ satisfying the Einstein
field equations, the CFEs encode a system of implied differential conditions for the curvature and
concomitants of the conformal factor associated with
$(\mathcal{M},\bmg)$ where $\bmg=\Xi^2\tilde{\bmg}$. As mentioned in the introduction, the key property
of these equations is that they are regular even at null infinity
$\mathscr{I}$, where $\Xi=0$ ---see also \cite{CFEbook} for a
comprehensive discussion.

\medskip

The so-called ``metric" version of the standard vacuum conformal Einstein field
equations are encoded in the following zero-quantities ---see
\cite{Fri81a,Fri81b,Fri82,Fri83}:
\begin{subequations}\label{CFE_tensor_zeroquants}
\begin{eqnarray}
&& Z_{ab} := \nabla_{a}\nabla_{b}\Xi +\Xi L_{ab} - s g_{ab}=0 ,
 \label{StandardCEFEsecondderivativeCF}\\
&& Z_{a} := \nabla_{a}s +L_{ac} \nabla ^{c}\Xi=0
 , \label{standardCEFEs}\\ && \delta_{bac} :=
 \nabla_{b}L_{ac}-\nabla_{a}L_{bc} - d_{abcd}\nabla^d{}\Xi =0
 , \label{standardCEFESchouten}\\ && \lambda_{abc}:=
 \nabla_{e}d_{abc}{}^{e}=0 , \label{standardCEFErescaledWeyl}\\ && Z
 := \lambda - 6 \Xi s + 3 (\nabla_{a}\Xi) \nabla^{a}\Xi,
\label{standardCFEconstraintFriedrichScalar}
\end{eqnarray}
\end{subequations}
where $\Xi$ is the conformal factor, $L_{ab}$ is the Schouten tensor,
defined in terms of the Ricci tensor $R_{ab}$ and the Ricci scalar $R$
via
\begin{equation}\label{SchoutenDefinition}
L_{ab}=\tfrac{1}{2}R_{ab}-\tfrac{1}{12}Rg_{ab},
\end{equation}
 $s$ is the so-called \emph{Friedrich scalar} defined as
\begin{equation}\label{s-definition}
s:= \tfrac{1}{4}\nabla_{a}\nabla^{a}\Xi + \tfrac{1}{24}R\Xi,
\end{equation}
and $d^{a}{}_{bcd}$ denotes the \emph{rescaled Weyl tensor}, defined
as
\[d^{a}{}_{bcd}=\Xi^{-1}C^{a}{}_{bcd},\]
where $C^{a}{}_{bcd}$ denotes the Weyl tensor.  The geometric meaning
of these zero-quantities is as follows. The equation $Z_{ab}=0$
encodes the conformal transformation law between ${R}_{ab}$ and
$\tilde{R}_{ab}$.  The equation $Z_{a}=0$ is obtained considering
$\nabla^{a}Z_{ab}$ and commuting covariant derivatives.  Equations
$\delta_{abc}=0$ and $\lambda_{abc}=0$ encode the contracted second
Bianchi identity. Finally, $Z=0$ is a constraint in the sense that if
it is verified at one point $p\in\mathcal{M}$ then $Z=0$ holds in
$\mathcal{M}$ by virtue of the previous equations.  A solution to the
metric conformal Einstein field equations consists of a collection of
fields
\[
\{g_{ab}, \; \Xi, \; s\;,L_{ab},\; d_{abcd}\}
\]
satisfying
\begin{equation}\label{vanishing_CFEs_tensorial_zq}
  Z_{ab}=0, \quad Z_{a}=0, \quad \delta_{abc}=0, \quad \lambda_{abc}=0, \quad Z=0.
\end{equation}

 \begin{remark}
   \emph{ If one opts to use the Ricci tensor $R_{ab}$ instead of the
     Schouten tensor $L_{ab}$ then the Ricci scalar $R$ appears in the
     right-hand side of equations but no equation for it has been
     provided.  In the CFEs the Ricci scalar encodes the
     \emph{conformal gauge source function}, hence there is no
     equation to fix that variable as it represents a gauge quantity.}
 \end{remark}

\noindent Since we are concerned here with spinor fields, we will need
the spinorial transcription of the CFEs (see \cite{CFEbook}), which reads
\begin{subequations}
\begin{eqnarray}
   && Z_{AA'BB'} = \nabla_{AA'}\nabla_{BB'}\Xi - \Xi \Phi _{ABA'B'} -
  s \epsilon _{AB} \epsilon _{A'B'} + \Xi \Lambda \epsilon _{AB}
  \epsilon _{A'B'} ,
  \label{Def_ConfFactor_CFE_zeroquant}\\
  && Z_{AA'} = \nabla_{AA'}s + \Lambda \nabla_{AA'}\Xi - \Phi
  _{ABA'B'} \nabla^{BB'}\Xi ,\label{Def_s_CFE_zeroquant}\\ &&
  \delta_{ABCC'} = \nabla_{A'(A}\Phi _{B)CC'}{}^{A'} - \epsilon _{C(A}
  \nabla_{B)C'}\Lambda + \phi _{ABCD} \nabla^{D}{}_{C'}\Xi
  ,\label{Def_delta_CFE_zeroquant} \\ && \Lambda _{C'ABC} =
  \nabla_{DC'}\phi _{ABC}{}^{D}, \label{Def_Lambda_CFE_zeroquant}\\ &&
  Z = \lambda -6 \Xi s + 3 (\nabla_{AA'}\Xi)
  \nabla^{AA'}\Xi, \label{Def_cons_CFE_zeroquant}
\end{eqnarray}
\end{subequations}
where $\Phi_{ABA'B'}$ and $\Lambda$ are as in section \ref{NotationAndSpinorFormalism} and the Weyl spinor enters via
the \emph{rescaled Weyl spinor}, $\phi_{ABCD}$, defined as
\begin{equation}\label{Def_rescaled_Weyl_spinor}
\phi_{ABCD} := \Xi^{-1} \Psi_{ABCD}
\end{equation}
---see \cite{Ste91, PenRin84} for more details.  As in the tensorial case, one
can choose the Schouten (tensor) spinor or the Ricci (tensor) spinor
as a variable.  
 \begin{remark}
 \label{G-C-M-Remark}
   \emph{In the initial value problem for the CFEs, $\phi_{ABCD}$ is
   determined by the initial data on a spacelike hypersurface
   $\mathcal{S}$. First decompose the rescaled Weyl spinor as
   \[ \phi_{ABCD}=E_{ABCD}+iB_{ABCD}\]
where $\bmE$ and $\bmB$ are the \textit{electric} and
\textit{magnetic} parts
   \[
   E_{ABCD} := \tfrac{1}{2}(\phi_{ABCD} + \hat{\phi}_{ABCD}),\qquad
   B_{ABCD} := \tfrac{i}{2}(-\phi_{ABCD} + \hat{\phi}_{ABCD}),
  \]
  with
  $\hat{\phi}_{ABCD}:=\tau_A{}^{A'}\tau_B{}^{B'}\tau_C{}^{C'}\tau_D{}^{D'}\bar{\phi}_{A'B'C'D'}$. The
  fields $\bmE,\bmB$ are the spinorial counterparts of the electric
  and magnetic parts of the rescaled Weyl tensor and comprise (part of
  the) initial data: away from $\mathscr{I}$ they are determined by
  a conformal analogue of the Gauss--Codazzi--Mainardi equations, while for the asymptotic
  initial value problem the constraint equations implied by the CFEs
  acquire a particularly simple form so that initial data for the
  magnetic part is determined algebraically by the Bach tensor of the
  induced metric $h_{ij}$ and the electric part must be prescribed,
  the only constraint being that it satisfies the
  \textit{TT} condition with respect to $h_{ij}$---see \cite{CFEbook,
    GasVal17a}.  For the discussion of this paper we will assume such
  data $\bmE, \bmB$ (and hence $\bm\phi$) to be given. Note also that one can
  formally define the Petrov type of an initial data set for the CFEs
  by applying the Petrov classification to the initial datum
  $\phi_{ABCD}|_{\mathcal{S}}$, rather than to $\Psi_{ABCD}$ ---see
  \cite{Ste91} for an introduction to the Petrov classification.  }
\end{remark}
The CFEs as previously presented can be regarded as a set of covariant
conditions for geometric fields on $(\mathcal{M},\bmg)$ and, hence,
they do not have a particular PDE character. However, there are
various hyperbolic reduction strategies, depending on the choice of
gauge fixing procedure, for extracting a set of evolution and
constraint equations.  For the subsequent discussion only the
evolution and constraint equations implied by the $\Lambda_{C'ABC}=0$
equation, namely,
\begin{equation}\label{RescaledWeylEquationDisplayed}
 \nabla^D{}_{C'}\phi _{ABCD}=0,
\end{equation}
will play an important role.  A direct calculation using the space
spinor formalism shows that equation
\eqref{RescaledWeylEquationDisplayed} can be recast as the following
system of evolution and constraint equations
\begin{align}\label{RescaledWeyl_evo_const}
  & \nabla_{\bm\tau} \phi _{ABCD} = 2 \mathcal{D}
  _{(A}{}^{F}\phi_{BCD)F}, \qquad \mathcal{D} ^{CD}\phi _{CDAB} = 0.
\end{align}
The CFEs can also be recast as a second-order system of wave
equations; see \cite{Pae13} for the tensorial formulation and
\cite{GasVal15} for the spinorial formulation.  We shall only need one
of these wave equations here, namely
\begin{eqnarray}
  \square \phi _{ABCF} = 12 \Lambda \phi _{ABCF} -6 \Xi \phi
  _{(AB}{}^{DG}\phi _{CF)DG},
  \label{Wave_eq_CFE_Weyl}
\end{eqnarray}
which is derived by considering $\nabla_D{}^{C'}\Lambda _{C'ABC}=0$ and
applying identity \eqref{DecomposeDoubleDerivativeContracted}.  It is
worth noting here that one of the tools used in \cite{GasVal15} to
show the equivalence between the system
\eqref{Def_ConfFactor_CFE_zeroquant}--\eqref{Def_cons_CFE_zeroquant}
and their wave-equation counterpart is the uniqueness property of
solutions to a certain class of \textit{homogeneous} wave equations, a
result which we shall also use repeatedly in this article and which is given below in
Theorem \ref{TheoremHomogeneousWave}.

\begin{definition}
{\em An operator $h$ is said to be \textit{homogeneous in} $\underline{u}$
\textit{and} $\partial\underline{u}$, if
$h (\mu\underline{u},\mu\partial\underline{u})=\mu
h(\underline{u}, \partial\underline{u})$ for all
$\mu\in\mathbb{C}$. }
\end{definition}

\begin{theorem}
\label{TheoremHomogeneousWave}
 Let $\mathcal{M}$ be a smooth manifold equipped with a Lorentzian
 metric $\bmg$ and consider the wave equation
\[\square \underline{u}=h (\underline{u},\partial\underline{u})\]
where $\underline{u}\in\mathbb{C}^m$ is a complex vector-valued
function on $\mathcal{M}$, $h:\mathbb{C}^{2m}\rightarrow\mathbb{C}^m$
is a smooth homogeneous function of its arguments and
$\square=g^{ab}\nabla_{a}\nabla_{b}$.  Let
$\mathcal{U}\subset\mathcal{S}$ be an open set and $\mathcal{S}\subset
\mathcal{M}$ be a spacelike hypersurface with normal $\tau^{a}$
respect to $\bmg$. Then the Cauchy problem
\begin{align*}
\square \underline{u}&=h (\underline{u},
\partial\underline{u}),\\ \underline{u}\left|_{\mathcal{U}}\right.&=\underline{u}_0,
\quad
\nabla_{\bm\tau}\underline{u}\left|_{\mathcal{U}}\right.=\underline{u}_1,
\end{align*} 
where $\underline{u}_{0}$ and $\underline{u}_{1}$ are smooth on
$\mathcal{U}$ and $\nabla_{\bm\tau}:= \tau^\mu\nabla_\mu$, has a
unique solution $\underline{u}$ in an open neighbourhood $\mathcal{W}$ of
$\mathcal{U}$, with $\mathcal{W} \subseteq\mathcal{D}^+(\mathcal{U})$.
\end{theorem}
We refer the reader to Proposition 3.2 of \cite{Tay96c} for a proof.

\begin{remark}{\em 
Analogous to the physical case, given a Petrov type D $\phi_{ABCD}$,
one can give an explicit construction of a Killing spinor, namely
\begin{equation*}
\kappa_{AB} = \phi^{-1/3}o_{(A}\iota_{B)}
\end{equation*}
in terms of the relevant adapted spin dyad $\lbrace \bmo,
\bm\iota\rbrace$. This can be seen directly by noting that the
equation \eqref{RescaledWeylEquationDisplayed} satisfied by the
rescaled Weyl spinor $\phi_{ABCD}$ is formally identical to the
physical vacuum Bianchi constraint and so the same computations as in
\cite{WalkerPenrose70} follow through.  An alternative approach to the
construction of Killing spinor initial data equations would be to
attempt to determine under what conditions the Petrov type of the
(rescaled) Weyl tensor restricted to $\mathcal{S}$ is propagated into
the spacetime development. Later, we shall see that such a result
follows, rather, as a \emph{product} of our analysis ---see Corollary
\ref{Corollary:PetrovPropagation}.}\label{Remark:DyadExpressionForKillingSpinorInTypeD}
\end{remark}

  \section{Conformal twistor initial data}
  \label{conformalTwistorKID}
  In this section, the conformal twistor initial data equations are
  derived. Although the main result of this article is the valence-2 (Killing spinor) case, the twistor case illustrates the main features of the calculation for the Killing spinor case (given in section \ref{conformalKSKID}) in a simpler setting.
\subsection{Twistor zero-quantities}
\label{Sec:TwistorZeroQuantities}

For the following discussion is convenient to define the following
\emph{zero-quantities}
\begin{subequations}
  \begin{eqnarray}
   && H_{A'AB} := 2
    \nabla_{A'(A}\kappa_{B)},\label{Def_H_twistor}\\ && B_{ABC}
    := \phi_{ABCD}\kappa^D.\label{Def_B_twistor}
    \end{eqnarray}
\end{subequations}
The spinors $H_{A'AB}$ and $B_{ABC}$ will be denoted in index free
notation as $\bmH$ and $\bmB$ and will be called the twistor
zero-quantity and the Buchdahl zero-quantity, respectively.  The
Buchdahl zero-quantity arises as an integrability condition of the
twistor equation.  To see this, notice that, taking the following
derivative of $\bmH$ and substituting definition
\eqref{Def_H_twistor}, one obtains
  \begin{equation}\label{curl_H_twistor}
  \nabla_{AA'}H^{A'}{}_{BC}= 2 \nabla_{AA'}\nabla_{(B}{}^{A'}\kappa
  _{C)} = \tfrac{1}{2} \epsilon _{AB} \square \kappa _{C}  +
  \tfrac{1}{2}  \epsilon _{AC} \square \kappa _{B} +
  \square_{BA}\kappa _{C} + \square_{CA}\kappa _{B}.
  \end{equation}
  Symmetrising and using equation \eqref{SpinorialRicciIdentities1} gives
  \[
  \nabla_{(A|A'|}H^{A'}{}_{BC)}= - 2\Psi_{ABCD}\kappa^D.
  \]
  The vanishing of the right-hand side of latter equation encodes the
  Buchdahl constraint, namely the fact that if $(\mathcal{M},\bmg)$
  admits a twistor then it is necessarily of Petrov type N or O. To
  write this in terms of the variables appearing in the CFEs,
  we substitute the definition of the rescaled Weyl spinor
  to obtain
  \begin{equation}\label{Curl_H_sym_toB_twistor}
  \nabla_{(A}{}^{A'}H_{|A'|BC)} = 2\Xi B_{ABC}.
  \end{equation}
  It is clear that if the unphysical spacetime
  $(\mathcal{M},\bmg)$ admits a twistor  then $H_{A'AB}=B_{ABC}=0$. 
  
\subsection{Auxiliary quantities and the twistor candidate equation}
  
The following auxiliary quantities
\begin{subequations}\label{def_twistor_aux_quants}
  \begin{eqnarray}
      && Q_{A}  := \nabla^{QA'}H_{A'QA}, \label{def_Q_twistor} \\
      && \xi_{A'} := \nabla^B{}_{A'}\kappa_B. \label{def_xi_twistor}
  \end{eqnarray}
\end{subequations}
will prove to be a useful bookkeeping device for the subsequent
calculations.
  The spinor $\xi_{A'}$ is merely a convenient
  placeholder for making irreducible decompositions of derivatives of
  $\kappa_A$:
  \begin{align}\label{decomp_Der_kappa}
    \nabla_{AA'}\kappa _{B} & = \tfrac{1}{2} \epsilon _{AB}
    \nabla_{CA'}\kappa ^{C} + \nabla_{(A|A'|}\kappa _{B)} = \tfrac{1}{2} H_{A'AB} - \tfrac{1}{2} \xi _{A'} \epsilon_{AB}.
  \end{align}
It is illustrative to introduce this
  shorthand since the analogous quantity in the Killing spinor case
  will play an important role in the calculation. On the other hand, the auxiliary quantity
  $Q_A$ will be central for the following discussion since it
  encodes a wave equation for $\kappa_A$. To see this, observe that
  tracing the identity \eqref{curl_H_twistor} and substituting 
  definition \eqref{def_Q_twistor} gives
\begin{equation}\label{Q_to_box_twistor_candidate}
Q_{A} = \tfrac{3}{2} \square \kappa _{A} + 3 \Lambda \kappa _{A}.
\end{equation}
Hence, $Q_{A}=0$ encodes the following wave equation for $\kappa_A$:
\begin{align} \label{Wave_eq_twistor_candidate}
\square \kappa_{A} + 2 \Lambda  \kappa_{A} =0.
\end{align}
A valence-1 spinor $\kappa_A$ satisfying \eqref{Wave_eq_twistor_candidate} will be called a \emph{twistor candidate}. To understand the
motivation for this definition and its name, notice that in general
any twistor trivially satisfies the twistor candidate
equation but the converse is not necessarily true:
\[
\bmH=0 \implies \bmQ =0, \qquad \text{but in general} \qquad \bmQ =0 \notimplies \bmH=0.
\]
However,  the initial data $(\kappa_A, \nabla_{\bm\tau}
\kappa_A)|_{\mathcal{S}}$ for the wave equation
\eqref{Wave_eq_twistor_candidate} have not yet been fixed. The present aim is to determine conditions on the twistor candidate initial data, on an initial hypersurface $\mathcal{S}$, which if propagated off $\mathcal{S}$ using
equation \eqref{Wave_eq_twistor_candidate} ensure that the corresponding twistor
candidate $\kappa_A$ is indeed a twistor ---i.e. such that
\begin{equation}
\bmQ =0 \;\;\&\;\; \text{twistor initial data} \;\;\implies \bmH=0.
\end{equation}
The strategy for obtaining such conditions on the initial data
$(\kappa_A, \nabla_{\bm\tau}
\kappa_A)|_{\mathcal{S}}$ will be to derive a closed
system of homogeneous wave equations for the zero-quantities $\bmH$
and $\bmB$ such that, if trivial initial data is given:
\begin{equation}\label{TrivialInitialData}
H_{A'AB}=0, \qquad \nabla_{\bm\tau} H_{A'AB}=0, \qquad B_{ABC}=0,
\qquad \nabla_{\bm\tau} B_{ABC}=0 \qquad \text{on} \qquad
\mathcal{U}\subset\mathcal{S}
\end{equation}
then Theorem \ref{TheoremHomogeneousWave} will guarantee that
$\bmH=0$ and $\bmB=0$ on some open neighbourhood $\mathcal{W}$ of $\mathcal{S}$, with
$\mathcal{W}\subseteq \mathcal{D}^{+}(\mathcal{S})$. Conditions \eqref{TrivialInitialData} will imply the desired restrictions, the \emph{conformal twistor initial data equations}, that must be satisfied by $\kappa_A$.  

\subsection{Wave equations for the zero-quantities}

To derive a wave equation for the zero-quantity $\bmH$, we start with the irreducible decomposition
\[
\nabla_{D}{}^{A'}H_{A'AB} = \tfrac{1}{3} \epsilon _{BD}
\nabla_{CA'}H^{A'}{}_{A}{}^{C} + \tfrac{1}{3} \epsilon _{AD}
\nabla_{CA'}H^{A'}{}_{B}{}^{C} + \nabla_{(A}{}^{A'}H_{|A'|BD)}.
\]
Substituting the definition \eqref{def_Q_twistor} and equation
\eqref{Curl_H_sym_toB_twistor}, it follows that
\begin{align}\label{derH_twistor_toBandQ}
\nabla_{D}{}^{A'}H_{A'AB} = 2 \Xi B_{ABD}  + \tfrac{1}{3} Q_{B}
\epsilon _{AD} + \tfrac{1}{3} Q_{A} \epsilon _{BD}.
\end{align}
Applying $\nabla_{D}{}^{B'}$ to the last expression, and using
identity \eqref{DecomposeDoubleDerivativeContracted} along with the
spinorial Ricci identities
\eqref{SpinorialRicciIdentities1}--\eqref{SpinorialRicciIdentities2},
renders
\begin{equation}\label{wave_H_twistor}
  \square H_{B'AB} = 6 \Lambda H_{B'AB} + 4 \Xi
  \nabla_{DB'}B_{AB}{}^{D} -4 B_{ABD} \nabla^{D}{}_{B'}\Xi -4
  \Phi_{(A}{}^{D}{}_{|B'}{}^{A'}H_{A'|B)D} + \tfrac{4}{3}
  \nabla_{(A|B'|}Q_{B)}.
\end{equation}
To derive a wave equation for $\bmB$, on the other hand, one begins by applying the D'Alembertian
operator $\square$ to the definition \eqref{Def_B_twistor} to obtain
\begin{align}\label{pre_wave_B_twistor}
\square B_{ABC} = \kappa ^{D} \square \phi _{ABCD} + \phi _{ABCD}
\square \kappa ^{D} + 2 (\nabla_{FA'}\phi _{ABCD}) \nabla^{FA'}\kappa
^{D}.
\end{align}
Substituting the definition \eqref{Def_B_twistor}, the identity
\eqref{Q_to_box_twistor_candidate}, and the wave equation satisfied by
the rescaled Weyl spinor \eqref{Wave_eq_CFE_Weyl} into the last
expression gives
\begin{equation}\label{wave_B_twistor}
\square B_{ABC} = 10 \Lambda B_{ABC} + H^{A'DF} \nabla_{FA'}\phi
_{ABCD} -6 \Xi B_{(A}{}^{DF}\phi _{BC)DF} + \tfrac{2}{3} \phi _{ABCD}
Q^{D}.
\end{equation}
Observe that if $Q_{A}=0$, namely if the twistor candidate wave
equation is imposed, then $\bmH$ and $\bmB$ satisfy the following set
of wave equations
\begin{subequations}
\begin{eqnarray}
  && \square H_{B'AB} = 6 \Lambda H_{B'AB} + 4 \Xi
  \nabla_{DB'}B_{AB}{}^{D}  -4 B_{ABD} \nabla^{D}{}_{B'}\Xi   -4 \Phi_{(A}{}^{D}{}_{|B'}{}^{A'}H_{A'|B)D},\qquad
   \label{Hom_wave_HandB1} \\
 && \square B_{ABC} = 10\Lambda B_{ABC} + H^{A'DF} \nabla_{FA'}\phi _{ABCD}  -6 \Xi B_{(A}{}^{DF}\phi
_{BC)DF}.  \label{Hom_wave_HandB2}
\end{eqnarray}
\end{subequations}
Notice that the only place where the CFEs (in their wave equation
form) were used is in substituting for the $\square \phi _{ABCD}$ term
in equation \eqref{pre_wave_B_twistor}.  \\

The important observation about equations
\eqref{Hom_wave_HandB1}--\eqref{Hom_wave_HandB2} is that they comprise
a closed system of \emph{regular and homogeneous} wave equations for
$\bmH$ and $\bmB$. Hence, we have the following:
\begin{proposition}\label{Prop:Propagation_twistor}
  Given initial data for the conformal Einstein field equations on
  $\mathcal{U}\subset\mathcal{S}$ where $\mathcal{S}$ is a spacelike
  hypersurface $\mathcal{S}$ with normal vector $\tau^{AA'}$, then a
  twistor candidate on $\mathcal{D}^{+}(\mathcal{U})$, is a true
  twistor (valence-1 Killing spinor) on an open neighbourhood
  $\mathcal{W}$ of $\mathcal{U}$, with
  $\mathcal{W}\subseteq\mathcal{D}^{+}(\mathcal{U})$, if and only if
\begin{subequations}
\begin{eqnarray}
  &&
  H_{A'AB}=B_{ABC}=0,\label{eq:VanishingOfH_twistor}\\ &&\nabla_{\bm\tau}H_{A'AB}=\nabla_{\bm\tau}
  B_{ABC}=0, \label{eq:VanishingOfNormalDerivB_twistor}
\end{eqnarray}
\end{subequations}
hold on $\mathcal{U}$.
\end{proposition}
\begin{proof}
The \emph{only if} direction is immediate. Suppose, on the other hand,
that $\kappa_A$ is a twistor candidate satisfying
\eqref{eq:VanishingOfH_twistor}--\eqref{eq:VanishingOfNormalDerivB_twistor}
on $\mathcal{U}\subset\mathcal{S}$.  As the zero-quantities
$H_{A'AB},~B_{ABC}$ satisfy the homogeneous wave equations
\eqref{Hom_wave_HandB1}--\eqref{Hom_wave_HandB2} then the uniqueness
result for homogeneous wave equations, given in Theorem
\ref{TheoremHomogeneousWave}, ensures that
\[ H_{A'AB}=0,\qquad B_{ABC}=0,\]
in an open neighbourhood $\mathcal{W}$ of $\mathcal{U}$, with
$\mathcal{W}\subseteq\mathcal{D}^{+}(\mathcal{U})$.  In other words, $\kappa_{A}$
solves the twistor equation in $\mathcal{W}$.
\end{proof}

\begin{remark}
\em{One important difference with \cite{GarVal08c}, in which the twistor equations are derived on a vacuum spacetime $(\tilde{\mathcal{M}},\tilde{\bmg})$, is that there the wave system
  closes with $\tilde{H}_{A'AB}$ alone and there is no need to
  introduce the analogous physical Buchdahl zero-quantity
  $\tilde{B}_{ABC}$.  Therefore it is interesting to check if in the
  conformal case one can also close the system with $H_{A'AB}$ alone.
  Observe that if one substitutes the expression for the Buchdahl
  constraint into equation \eqref{Hom_wave_HandB1} and uses the CFEs, then 
  one obtains
\begin{multline}\label{WaveH_twistor_singular}
  \square H_{A'AB} = - 2\Xi^{-1} (\nabla^{C}{}_{A'}\Xi)
  \nabla_{(A}{}^{B'}H_{|B'|BC)}+ 6 \Lambda H_{A'AB} \\-2 \Xi \phi
  _{ABCD}H_{A'}{}^{CD} -4 \Phi_{(A}{}^{C}{}_{|A'}{}^{B'}H_{B'|B)C}.
 \end{multline}
Hence, $\bmH$ satisfies a closed and homogeneous, though \emph{singular}, equation
due to the $\Xi^{-1}$ coefficient. Theorem
\ref{TheoremHomogeneousWave} does not apply in this case.  From equation
\eqref{WaveH_twistor_singular} one can recover the analogous wave
equation in the physical case discussed in \cite{GarVal08c} simply by
adding a tilde to the fields and setting $\Xi=1$. Arguably, one could try to use the theory of \textit{Fuchsian
  systems}, as used in \cite{ChrPaetz13,Pae14a}, to see if the
analogue of Theorem \ref{TheoremHomogeneousWave} applies for the
singular equation \eqref{WaveH_twistor_singular}.  However, one of the
advantages of the conformal approach of the CFEs is that one deals
with manifestly regular equations.  Therefore, from this perspective,
it is preferable to deal with manifestly regular equations by introducing $B_{ABC}$ as a further
zero-quantity to be propagated.  The same observation holds for the
conformal valence-2 Killing spinor initial data discussion of the
following sections, where, to close the system in a regular way, one
needs to introduce not only a ``Buchdahl" zero-quantity but also a
further derivative thereof.}
\end{remark}

\subsection{Intrinsic conformal twistor initial data conditions}
\label{Sec:IntrinsicTwistor}

Proposition \ref{Prop:Propagation_twistor} of the previous section
gives necessary and sufficient conditions for a twistor candidate to
correspond to a true twistor. We would like now to reduce the
conditions
\eqref{eq:VanishingOfH_twistor}--\eqref{eq:VanishingOfNormalDerivB_twistor},
which contain not only derivatives tangential to $\mathcal{S}$ but
also normal to it, to conditions on the field
$\kappa_A|_{\mathcal{U}}$ that are computable at the level of initial
data for the CFEs. More precisely, our twistor candidate will be constructed as the
solution to the following initial value problem:
\begin{equation}\label{TwistorIVP}
    \left\{
\begin{array}{ll}
	 \square \kappa_A + 2\Lambda \kappa_A=0 & \qquad
         \text{on}~\mathcal{D}^{+}(\mathcal{U}),\\ \kappa_A =
         \bar{\kappa}_A
         &\qquad\text{on}~\mathcal{U},\\ \nabla_{\bm\tau} \kappa_{A} +
         \tfrac{2}{3}\mathcal{D}_{A}{}^{B}\kappa_{B}=0
         &\qquad\text{on}~\mathcal{U},
\end{array} \right.
\end{equation}
where the initial data $\bar{\kappa}_A$ will be appropriately
restricted in order to ensure that the conditions
\eqref{eq:VanishingOfH_twistor} and
\eqref{eq:VanishingOfNormalDerivB_twistor} hold for the solution
$\kappa_A$. Note that the bulk equation here is precisely $Q_A=0$, the
twistor candidate equation.  \\

We begin by introducing the following definitions:
\begin{align}
  \mathcal{H} _{ABC} := \tau _{(A}{}^{A'}H_{|A'|BC)}, \qquad
  \mathcal{H}_{A} := \tau^{QA'} H_{A'AQ},
\end{align}
in terms of which the space spinor split of $H_{A'AB}$ reads
\begin{align}
  H_{A'AB} = - \tfrac{1}{2} \tau ^{C}{}_{A'} \mathcal{H} _{ABC} +
  \tfrac{1}{6} \tau_{AA'} \mathcal{H}_{B} + \tfrac{1}{6} \tau_{BA'}
  \mathcal{H}_{A}.
\end{align}
Note the space spinors $\mathcal{H} _{ABC}$ and $\mathcal{H}_{A}$
contain all the information of $H_{A'AB}$; in particular,
\[
H_{A'AB}=0 \quad                 
\iff \quad \mathcal{H} _{A}=0    
\quad
\& \quad \mathcal{H}_{ABC}=0.  
\]
Substituting the definition of $\bmH$, equation \eqref{Def_H_twistor},
one obtains
\begin{align}\label{spacespinordecompHtotwistorders}
\mathcal{H} _{A} \equiv \tfrac{3}{2} \nabla_{\bm\tau} \kappa_{A} +
\mathcal{D} _{A}{}^B\kappa_{B}, \qquad \mathcal{H} _{ABC} \equiv 2
\mathcal{D} _{(AB}\kappa _{C)},
\end{align}
Notice that $\mathcal{H}_A|_{\mathcal{U}}=0$ is precisely the second
initial condition of \eqref{TwistorIVP}. The condition
$\mathcal{H}_{ABC}=0$, on the other hand, contains only quantities
intrinsic to $\mathcal{S}$, which is also true of the condition
$B_{ABC}=0$ from equation \eqref{eq:VanishingOfH_twistor}.
This motivates the following definition:

\begin{definition}{\em
A spinor field $\bar{\kappa}_A$ defined on some
$\mathcal{U}\subset\mathcal{S}$ and satisfying
\begin{equation}
    \mathcal{H}(\bar{\bm\kappa})_{ABC}\equiv
    2\mathcal{D}_{(AB}\bar{\kappa}_{C)}=0,\qquad
    B(\bar{\bm\kappa})_{ABC}\equiv
    \phi_{ABCD}\bar{\kappa}^D=0 \label{twistor_CSKIDs}
\end{equation}
will be called a \emph{conformal twistor initial data set} on
$\mathcal{U}$.}
\end{definition}
We will show that a conformal twistor initial dataset,
$\bar{\kappa}_A$, indeed comprises initial data for a spacetime
twistor, in that the resulting solution of the initial value problem \eqref{TwistorIVP}
necessarily satisfies the twistor equation on an open neighbourhood $\mathcal{W}$
of $\mathcal{U}$, with $\mathcal{W}\subseteq\mathcal{D}^+(\mathcal{U})$.
First we establish the following:
\begin{lemma}\label{lemma:twistor_removing_redundancy}
  Given initial data for the conformal Einstein field equations on
  $\mathcal{U}\subset\mathcal{S}$, where $\mathcal{S}$ is a spacelike
  hypersurface with normal $\bm\tau$, if $\bar{\kappa}_A$ is a
  conformal twistor initial data set on $\mathcal{U}$, then the
  solution $\kappa_A$ of the initial value problem \eqref{TwistorIVP}
  satisfies
\[H(\bm\kappa)_{A'AB}=B(\bm\kappa)_{ABC}=\nabla_{\bm\tau}H(\bm\kappa)_{A'AB}=\nabla_{\bm\tau}B(\bm\kappa)_{ABC}=0 \]
on $\mathcal{U}$.
\end{lemma}
\begin{proof}
 The assumption that $\bar{\kappa}_A$ satisfies
 $\mathcal{H}(\bar{\bm\kappa})_{ABC}=0$ implies that the solution
 of \eqref{TwistorIVP} also satisfies
\[\mathcal{H}_{ABC}|_{\mathcal{U}}=0, \qquad \mathcal{H}_A|_{\mathcal{U}}=0,\]
the latter following from the second initial condition of the initial value problem
\eqref{TwistorIVP}. Hence, as remarked above, we have
$H_{A'ABC}|_{\mathcal{U}}=0$. Clearly $B_{ABC}|_{\mathcal{U}}=0$ also,
since $\kappa_A|_{\mathcal{U}}=\bar{\kappa}_A$. Now, substituting the
space spinor split of $\nabla$ in the identity
\eqref{derH_twistor_toBandQ}, it follows that
\begin{align*}
  \tau _{D}{}^{A'}\nabla_{\bm\tau} H_{A'AB} -2 \tau ^{CA'} \mathcal{D}
  _{DC}H_{A'AB} = 4\Xi B_{ABD} + \tfrac{4}{3} Q_{(A}\epsilon _{B)D}.
\end{align*}
Transvecting with $\tau^{D}{}_{B'}$ and rearranging gives
\begin{align}\label{time-derHToBHQ}
\nabla_{\bm\tau} H_{B'AB} = -4 \Xi B_{ABD} \tau ^{D}{}_{B'} -2 \tau
^{CA'} \tau ^{D}{}_{B'} \mathcal{D} _{DC}H_{A'AB} - \tfrac{4}{3} \tau
^{D}{}_{B'}Q_{(A}\epsilon _{B)D}.
\end{align}
Now, since $\kappa_A$ satisfies \eqref{TwistorIVP}, we have in
particular that $Q_A=0$ and hence the conditions
$H_{A'ABC}|_{\mathcal{U}} = B_{ABC}|_{\mathcal{U}}=0$ also imply that
\[\nabla_{\bm\tau} H_{A'AB}|_{\mathcal{U}}=0.\] 
All that remains, then, is to establish that
$\nabla_{\tau}B_{ABC}|_{\mathcal{U}}=0$. Using the definition of
$\bmB$, \eqref{Def_B_twistor}, we find
\begin{align*}
\nabla_{\bm\tau} B_{ABC} = \phi _{ABCD}\nabla_{\bm\tau} \kappa ^{D} +
\kappa ^{D} \nabla_{\bm\tau} \phi _{ABCD} .
\end{align*}
At this point one can exploit the evolution equation for
$\phi_{ABCD}$, namely \eqref{RescaledWeyl_evo_const}, along with the
initial conditions of \eqref{TwistorIVP}, to obtain
\begin{align}\label{normalderB_twistor_exp}
\nabla_{\tau}B_{ABC}|_{\mathcal{S}}= -2\kappa ^{D} \mathcal{D}
_{DF}\phi _{ABC}{}^{F} + \tfrac{2}{3} \phi _{ABCD} \mathcal{D}
^{D}{}_{F}\kappa ^{F} + \tfrac{2}{3}\phi_{ABCD}\mathcal{H}^D .
\end{align}
In fact, the right-hand-side can be completely rewritten in terms of
$\mathcal{H}_A|_{\mathcal{S}}$, $\mathcal{H}_{ABC}|_{\mathcal{S}}$ and
$B_{ABC}|_{\mathcal{U}}$. To do so, first swap indices $D$ and $A$ in
the first term on the right-hand-side of equation
\eqref{normalderB_twistor_exp}, which is made possible by the
constraint equation for the rescaled Weyl spinor, equation
\eqref{RescaledWeyl_evo_const}, to obtain
\begin{align}\label{normalderB_twistor_exp2}
\nabla_{\tau}B_{ABC}|_{\mathcal{U}}= -2 \kappa ^{D} \mathcal{D}
_{AF}\phi _{DBC}{}^{F} + \tfrac{2}{3} \phi _{ABCD} \mathcal{D}
^{D}{}_{F}\kappa ^{F} + \tfrac{2}{3}\phi_{ABCD}\mathcal{H}^D.
\end{align}
Substituting the definition of $B_{ABC}$ into the expression
$\mathcal{D} _{AD}B_{BC}{}^{D}$ and using the Leibnitz rule, one sees
that the above as equivalent to
\begin{equation}\label{normalderB_twistor_exp3}
\nabla_{\tau}B_{ABC}|_{\mathcal{U}}= -2 \mathcal{D} _{AD}B_{BC}{}^{D}
-2 \phi _{BCDF} \mathcal{D} _{A}{}^{F}\kappa ^{D} +\tfrac{2}{3} \phi
_{ABCF} \mathcal{D} _{D}{}^{F}\kappa ^{D} +
\tfrac{2}{3}\phi_{ABCD}\mathcal{H}^D.
\end{equation}
Now, performing the irreducible decomposition of $\mathcal{D}
_{AB}\kappa _{C}$ and using the expression for $\mathcal{H}_{ABC}$ in
\eqref{spacespinordecompHtotwistorders}, one has
\begin{equation}\label{decompSenKappa}
\mathcal{D} _{AB}\kappa _{C} = \tfrac{1}{2} \mathcal{H} _{ABC} +
\tfrac{1}{3} \epsilon _{BC} \mathcal{D} _{AD}\kappa ^{D} +
\tfrac{1}{3} \epsilon _{AC} \mathcal{D} _{BD}\kappa ^{D}.
\end{equation}
Finally, substituting decomposition \eqref{decompSenKappa} into
equation \eqref{normalderB_twistor_exp3} gives
\begin{align}\label{time-derBToBH}
\nabla_{\tau}B_{ABC}|_{\mathcal{U}}=- \phi _{BCDF} \mathcal{H}
_{A}{}^{DF} -2 \mathcal{D} _{AD}B_{BC}{}^{D} +
\tfrac{2}{3}\phi_{ABCD}\mathcal{H}^D.
\end{align}
Hence it follows that if $\mathcal{H}_{ABC}=\mathcal{H}_A=B_{ABC}=0$
on $\mathcal{U}$, then
\[\nabla_{\bm\tau}B_{ABC}|_{\mathcal{U}}=0, \]
and the conclusion follows.
\end{proof} 
\begin{remark}\emph{
Observe that although the CFEs were not needed in deriving equation \eqref{time-derHToBHQ}, they were however used when deriving \eqref{time-derBToBH}, so we have made crucial use of the assumption of having initial data satisfying the CFE
constraints in Lemma \ref{lemma:twistor_removing_redundancy}.  }
\end{remark}

We then obtain the main theorem of this section as a simple
application of Lemma \ref{lemma:twistor_removing_redundancy}:

\begin{theorem}\label{Theorem_twistor}
Consider an initial data set for the vacuum conformal Einstein field
equations, as encoded in the CFE zero-quantities
\eqref{Def_ConfFactor_CFE_zeroquant}--\eqref{Def_cons_CFE_zeroquant},
on a spacelike hypersurface $\mathcal{S}$ and let
$\mathcal{U}\subset\mathcal{S}$ be an open set.  The development of
the initial data set will have a twistor (valence-1 Killing spinor) in
an open neighbourhood $\mathcal{W}$ of $\mathcal{U}$, with
$\mathcal{W}\subseteq \mathcal{D}^{+}(\mathcal{U})$, if and only if
there exists a conformal twistor initial data set $\bar{\kappa}_A$ on
$\mathcal{U}$.  Given the existence of such a $\bar{\kappa}_A$, the
twistor $\kappa_A$ is obtained as the solution of the initial value
problem \eqref{TwistorIVP}.
\end{theorem}
\begin{proof}
  Lemma \ref{lemma:twistor_removing_redundancy} implies that if
  $\bar{\kappa}_A$ is a conformal twistor initial data set, then
  $\kappa_A$ satisfies
  \eqref{eq:VanishingOfH_twistor}--\eqref{eq:VanishingOfNormalDerivB_twistor},
  then, by virtue of Proposition \ref{Prop:Propagation_twistor}, one
  has that $H_{A'ABC}=0$ in an open neighbourhood $\mathcal{W}$ of
  $\mathcal{U}$, with $\mathcal{W}\subseteq
  \mathcal{D}^{+}(\mathcal{U})$. Hence, $\kappa_A$ is a
  twistor on $\mathcal{W}$.
\end{proof}

\section{Conformal Killing spinor initial data}
  \label{conformalKSKID}

In this section we perform the valence-2 counterpart of the analysis in previous section: we derive necessary and sufficient conditions for a spinor field defined on (a subset of) an initial hypersurface $\mathcal{S}$ to give rise to a Killing spinor on the spacetime development. These conditions will be given by the so-called \emph{conformal Killing spinor initial data
equations}, or \emph{CKSIDs} for short. Although the calculations will be more
involved, in general terms we follow the same strategy as in the twistor case.

\subsection{Killing spinor zero-quantities}

Analogous to the twistor case, we define the zero-quantities:
\begin{equation}\label{KS_zero_quantities1}
H_{A'ABC}:=3\nabla_{A'(A}\kappa_{BC)}, \qquad
B_{ABCD}:=\kappa_{(A}{}^Q\phi_{BCD)Q}.
\end{equation}
A short computation shows that
\begin{equation}
\nabla_{(A}{}^{A'}H_{\vert A'\vert BCD)} = 6\Xi
B_{ABCD}. \label{Eq:BuchdahlAsCurlOfH}
\end{equation}
Despite the formal resemblence with the equations of section
\ref{Sec:TwistorZeroQuantities}, the following discussion will show
that, unlike the twistor case, one cannot obtain a closed homogeneous
wave system in terms of these variables alone; we shall need the
following additional zero-quantity
\begin{equation}
 F_{A'BCD}:=\nabla^Q{}_{A'}B_{QBCD}.\label{Eq:DefZeroQuantityF}
\end{equation}
We shall show that it is however possible to derive a closed
homogeneous wave system for the fields $(\bmH, \bmB, \bmF)$.

\subsection{Auxiliary quantities and the Killing spinor candidate equation}

By analogy with the twistor case, it will prove useful to define the
following auxiliary quantity
\begin{equation}
  Q_{BC} := \tfrac{1}{2}\nabla^{AA'}H_{A'ABC}.
 \label{Eq:WaveForKS_DefinitionQ}
\end{equation}
A direct calculation shows that the auxiliary quantity $Q_{AB}=0$
encodes a wave equation for $\kappa_{AB}$:
\begin{equation}
  Q_{BC} \equiv \square \kappa_{BC} + 4 \Lambda\kappa_{BC} - \Xi
  \phi_{BCAD}\kappa^{AD}.
  \label{Eq:WaveForKS}
\end{equation}
A solution of $Q_{AB}=0$ will be called a \emph{Killing spinor candidate}. It will also prove convenient to define the
auxiliary spinor $\xi_{AA'} := \nabla^{B}{}_{A'}\kappa_{AB}$ in terms
of which one can perform the following decomposition
\begin{equation}
\nabla_{AA'}\kappa_{BC} = \tfrac{1}{3} H_{A'ABC} - \tfrac{1}{3}
\xi_{CA'} \epsilon_{AB} - \tfrac{1}{3} \xi_{BA'}
\epsilon_{AC}.\label{Eq:DecompGradKS}
\end{equation}
Using the latter expression, one can show by a straightforward
computation that
\begin{equation}\label{eq:non-KillingVectorExplained}
  \nabla^Q{}_{A'}H_{B'ABQ} + Q_{AB} \epsilon_{A'B'} =
  \nabla_{AA'}\xi_{BB'} + \nabla_{BB'}\xi_{AA'} + 6 \kappa_{(A}{}^{Q}
  \Phi_{B)QA'B'},
\end{equation}
which in turn implies the following identity
\begin{multline}
\nabla_{AA'}\xi_{BB'} = - \tfrac{1}{2} \epsilon_{AB}
\nabla^{C}{}_{(A'}\xi_{|C|B')}- 3 \kappa_{(A}{}^{C} \Phi_{B)CA'B'} - 3
\Lambda \kappa_{AB} \epsilon_{A'B'} \\+ \tfrac{3}{4} \Xi \kappa^{CD}
\phi_{ABCD} \epsilon_{A'B'} + \tfrac{1}{4} Q_{AB} \epsilon_{A'B'} +
\tfrac{1}{2} \nabla^Q{}_{(A'}H_{B')ABC},\label{Eq:DecompGradXi}
\end{multline}
which will prove useful later. On the other hand, it is
straightforward to show using equations \eqref{Eq:BuchdahlAsCurlOfH}
and \eqref{Eq:WaveForKS_DefinitionQ} that
\begin{equation}
     \nabla_{D}{}^{A'}H_{A'ABC} = 6 \Xi B_{ABCD} + \tfrac{3}{2}
     \ Q_{(AB}\epsilon_{C)D}.\label{IrrDecompCurlOfH}
\end{equation}

\begin{remark}
\label{GeometricSignificance}
  \emph{ As stressed in section \ref{sec:Introduction}, in general
  the auxiliary spinor $\xi_{AA'}$ is not the spinorial counterpart of
  a Killing vector. Contrast this with the case of the physical
  framework ---namely $(\tilde{\mathcal{M}},\tilde{\bmg})$ satisfying
  the vacuum Einstein field equations--- where the last term in
  equation \eqref{eq:non-KillingVectorExplained} vanishes and hence
  $\tilde{\bm\xi}$ is a Killing vector. This point is subtle even
  in the physical framework if one departs from the vacuum case: for instance, if one considers matter models such as
  Einstein-Maxwell then it is necessary to make further assumptions
  such as the \emph{matter alignment condition} to ensure that $\tilde{\bm\xi}$ is a
  Killing vector ---see \cite{ValCol16} for details. This property of
  $\tilde{\bm\xi}$ is crucial for the derivation of the physical
  Killing spinor data equations presented in \cite{GarVal08c} and
  \cite{ValCol16}, which involves the Killing vector zero-quantity $\tilde{S}_{ab}:=
  \tilde{\nabla}_{(a}\tilde{\xi}_{b)}$. 
  In the unphysical
  framework one cannot appeal to this strategy since the unphysical
  Ricci spinor $\bm\Phi$ is non-vanishing and does not satisfy any
  useful algebraic relation.
  As we will see in
  the following, the key to solving this problem in the unphysical
  framework is not to introduce the analogous quantity $\bmS$ 
  but, instead, to
  focus on the Buchdahl constraint $\bmB$ and its derivative $\bmF$. Note however that the quantity 
  \begin{equation}\label{eq:conformalKillingVector}
X_{AA'}=\Xi \xi_{AA'} - 3 \kappa_{AQ}\nabla_{A'}{}^{Q}\Xi,
\end{equation}
which we do not make explicit use of, does have geometric significance in that it is a conformal Killing vector on $(\mathcal{M},\bmg)$. Indeed, a short computation verifies that
\begin{multline}\label{conformalKillingvector}
  \nabla_{AA'}X_{BB'}+\nabla_{BB'}X_{AA'}-\tfrac{1}{2}\epsilon_{AB}\epsilon_{A'B'}\nabla^{CC'}X_{CC'}
  =-\Xi \nabla^Q{}_{(A'}H_{B')ABQ} - 2\nabla^C{}_{(A'}H_{B')ABC}
\end{multline}
and moreover that $X_{AA'}=\Xi^2 \tilde{\xi}_{AA'}$, where $\tilde{\xi}_{AA'}$ is the Killing vector associated to $\tilde{\kappa}_{AB}=\Xi^{-2}\kappa_{AB}$ in the physical spacetime (with $\tilde{\bmg}=\Xi^{-2}\bmg$). 
  }
\end{remark}


\subsection{Wave equations for the zero-quantities}
\label{Sec:KSWaveEqs}

Applying $\nabla^A{}_{B'}$ to equation \eqref{IrrDecompCurlOfH}, one
obtains the following wave equation:
\begin{multline}
    \square H_{B'ABC} = 6 \Lambda H_{B'ABC} - 12 \Xi F_{B'ABC} -
    12(\nabla^{D}{}_{B'}\Xi) B_{ABCD} \\+ 3\nabla_{(A|B'|}Q_{BC)} - 6
    \Phi_{(A}{}^{D}{}_{|B'}{}^{A'}H_{A'|BC)D}. \label{Eq:WaveEqForH}
\end{multline}
Similarly, applying $\nabla_A{}^{A'}$ to equation
\eqref{Eq:DefZeroQuantityF}, it is straightforward to verify the
following wave equation for $B_{ABCD}$:
\begin{equation}
     \square B_{ABCD} = 12\Lambda B_{ABCD} - 6\Xi
     \phi_{(AB}{}^{FG}B_{CD)FG} +
     2\nabla_{AA'}F^{A'}{}_{BCD}. \label{Eq:FirstWaveEqForB}
\end{equation}
The task remaining is to derive a wave equation for $F_{A'ABC}$.  To
do so, we will need some ancillary identities. Firstly, a direct
calculation shows that
\begin{multline}
2\phi_{(AB}{}^{GH}B_{CF)GH}=
\kappa_A{}^D\phi_{(BC}{}^{GH}\phi_{FD)GH}+\kappa_B{}^D\phi_{(AC}{}^{GH}\phi_{FD)GH}
\\+\kappa_C{}^D\phi_{(AB}{}^{GH}\phi_{FD)GH}+\kappa_F{}^D\phi_{(BC}{}^{GH}\phi_{AD)GH}.
\label{Eq:UsefulIdentity1}
 \end{multline}
This identity, along with the irreducible decomposition
\begin{align*}
\phi_{ABCD} \phi_{FGH}{}^{D} &= \tfrac{1}{24} \phi_{DLMP} \phi^{DLMP}
\epsilon_{AH} \epsilon_{BG} \epsilon_{CF} + \tfrac{1}{24} \phi_{DLMP}
\phi^{DLMP} \epsilon_{AG} \epsilon_{BH} \epsilon_{CF} \\ & +
\tfrac{1}{24} \phi_{DLMP} \phi^{DLMP} \epsilon_{AH} \epsilon_{BF}
\epsilon_{CG} + \tfrac{1}{24} \phi_{DLMP} \phi^{DLMP} \epsilon_{AF}
\epsilon_{BH} \epsilon_{CG} \\ &+ \tfrac{1}{24} \phi_{DLMP}
\phi^{DLMP} \epsilon_{AG} \epsilon_{BF} \epsilon_{CH} + \tfrac{1}{24}
\phi_{DLMP} \phi^{DLMP} \epsilon_{AF} \epsilon_{BG} \epsilon_{CH}
\\ &+ \tfrac{1}{6} \epsilon_{CH} \phi_{(AB}{}^{DL}\phi_{FG)DL} +
\tfrac{1}{6} \epsilon_{CG} \phi_{(AB}{}^{DL}\phi_{FH)DL} +
\tfrac{1}{6} \epsilon_{CF} \phi_{(AB}{}^{DL}\phi_{GH)DL}\\ & +
\tfrac{1}{6} \epsilon_{BH} \phi_{(AC}{}^{DL}\phi_{FG)DL} +
\tfrac{1}{6} \epsilon_{BG} \phi_{(AC}{}^{DL}\phi_{FH)DL} +
\tfrac{1}{6} \epsilon_{BF} \phi_{(AC}{}^{DL}\phi_{GH)DL}\\ &+
\tfrac{1}{6} \epsilon_{AH} \phi_{(BC}{}^{DL}\phi_{FG)DL} +
\tfrac{1}{6} \epsilon_{AG} \phi_{(BC}{}^{DL}\phi_{FH)DL} +
\tfrac{1}{6} \epsilon_{AF} \phi_{(BC}{}^{DL}\phi_{GH)DL},
\end{align*}
allows one to derive the following identity:
\begin{align}
  \kappa^{DG}\phi_{(ABC}{}^H\phi_{F)HDG}&=2\phi_{(AB}{}^{GH}B_{CF)GH}.
  \label{Eq:UsefulIdentity2}
\end{align}
Now, using the definition of the Buchdahl zero-quantity
\eqref{KS_zero_quantities1}, the CFEs for $\bm\phi$ in its first and
second order form, namely equations
\eqref{RescaledWeylEquationDisplayed} and \eqref{Wave_eq_CFE_Weyl},
and using the decomposition \eqref{Eq:DecompGradKS}, we get
\begin{multline} 
  \square B_{ABCD} = \tfrac{2}{3} \nabla_{\bm\xi}\phi_{ABCD} + 8
  \Lambda B_{ABCD} - \phi_{(ABC}{}^{F}Q_{D)F} + \tfrac{2}{3}
  H^{A'}{}_{(A}{}^{FG}\nabla_{|FA'|}\phi_{BCD)G} \\ - 6 \Xi
  \kappa_{(A}{}^{F}\phi_{BC}{}^{GH}\phi_{D)FGH} - \Xi
  \kappa^{FG}\phi_{(ABC}{}^{H}\phi_{D)FGH}.
\end{multline}
where $\nabla_{\bm\xi} := \xi^{AA'}\nabla_{AA'}$. Substituting the
above identities
\eqref{Eq:UsefulIdentity1}--\eqref{Eq:UsefulIdentity2}, we can derive
the following alternative (non-homogeneous) wave equation for
$B_{ABCD}$:
\begin{multline}
    \square B_{ABCD} = \tfrac{2}{3}\nabla_{\bm\xi}\phi_{ABCD} +
    8\Lambda B_{ABCD} - 14\Xi \phi_{(AB}{}^{FG}B_{CD)FG} \\ +
    \tfrac{2}{3}(\nabla_{FA'}\phi_{G(ABC})H^{A'}{}_{D)}{}^{FG} -
    \phi_{(ABC}{}^{F}Q_{D)F}. \label{Eq:SecondWaveEqForB}
\end{multline}
\begin{remark}
  \emph{ Equation \eqref{Eq:SecondWaveEqForB} actually encodes the
  fact that, \emph{given a Killing spinor} $\kappa_{AB}$, the field
  $\xi_{AA'}$ is a collineation for the rescaled Weyl tensor
  ---i.e. that
\[  \mathcal{L}_{\bm\xi}\phi_{ABCD} := \nabla_{\bm\xi}\phi_{ABCD} + \phi_{F(ABC}\nabla_{D)A'}\xi^{FA'}=0,\] 
the definition\footnote{The Lie derivative does not extend to spinor
fields, in general. See \cite{PenRin86} for a discussion.} of the
``Lie derivative" being derived from the spinorialised counterpart of
$\mathcal{L}_{\bm\xi}d_{abcd}$. Indeed, this fact follows from a
straightforward calculation
\begin{align*}
    \mathcal{L}_{\bm\xi}\phi_{ABCD}&:= \nabla_{\bm\xi}\phi_{ABCD}
    + \phi_{F(ABC}\nabla_{D)A'}\xi^{FA'} \\ &=
    \nabla_{\bm\xi}\phi_{ABCD}-6\Lambda \kappa_{(D}{}^{F} \phi_{ABC)F}
    - \tfrac{3}{2} \Xi \ \kappa^{FG} \phi_{(ABC}{}^{H} \phi_{D)FGH} +
    \tfrac{1}{4} \phi_{F(ABC} \ Q_{D)}{}^{F}\\ &=
    \nabla_{\bm\xi}\phi_{ABCD}-6\Lambda B_{ABCD} - 3\Xi
    \phi_{(AB}{}^{FG}B_{CD)FG} + \tfrac{1}{4}\phi_{F(ABC}Q_{D)}{}^F,
\end{align*}
where we are using \eqref{Eq:DecompGradXi} along with the identity
\eqref{Eq:UsefulIdentity2}. Given a Killing spinor $\kappa_{AB}$, the
resulting zero-quantities vanish: $B_{ABCD}=H_{A'ABC}=Q_{AB}=0$, and
hence it follows from \eqref{Eq:SecondWaveEqForB} that
\begin{equation}\mathcal{L}_{\bm\xi}\phi_{ABCD} = \nabla_{\bm\xi}\phi_{ABCD} =0. \label{Collineation}
\end{equation}
Thus, a sub-product of our analysis is that if $(\mathcal{M},\bmg)$
admits a Killing spinor, then the Weyl-collineation condition
\eqref{Collineation} is satisfied.  This is not trivial since, as remarked above, the
vector $\xi_{AA'}$ is not in general a
conformal Killing vector. In contrast,
for the physical spacetime it is clear that this condition holds 
since in that case $\tilde{\xi}_{AA'}$ is a Killing vector and thus
$\mathcal{L}_{\bm{\tilde{\xi}}}C_{abcd}=0$, trivially.
}
\end{remark}

The key observation to close the system is that there are no derivatives of zero-quantities
appearing on the right-hand-side of equation \eqref{Eq:SecondWaveEqForB}. This,
combined with the CFE $\Lambda_{A'ABC}=0$, equation \eqref{RescaledWeylEquationDisplayed},
suggests
that by applying $\nabla^A{}_{A'}$ to equation
\eqref{Eq:SecondWaveEqForB} we may be able to derive a wave equation
for $F_{A'BCD}$ with the desired properties, namely being homogeneous
in $(\bmH, \bmB, \bmF)$ and their first derivatives (apart from the
D'Alembertian term). We will see that this strategy does indeed work;
the difficult part is in deriving a suitable expression for
$\nabla^A{}_{A'}\nabla_{\bm\xi}\phi_{ABCD}$ in terms of the zero-quantities. To do so, first, we first note some further useful identities: from the definition
of the zero-quantities $\bmH$, $\bmB$ and the auxiliary spinor
$\bm\xi$, we obtain
\begin{align}
    \kappa_{A}{}^{F} \nabla_{FF'}\phi_{BCDG} &=\kappa_{A}{}^{F}
    \nabla_{DF'}\phi_{BCGF} \nonumber\\ &= \tfrac{1}{3} \xi_{AF'}
    \phi_{BCDG} - \tfrac{1}{3} \xi^{F}{}_{F'} \phi_{BCGF}
    \epsilon_{AD} - \tfrac{1}{3}
    \xi^{F}{}_{F'}\phi_{(BC|DF}\epsilon_{A|G)} \nonumber \\ & -
    \tfrac{2}{3} \xi^{F}{}_{F'}\phi_{(B|D|C|F}\epsilon_{A|G)} +
    \epsilon_{AD} F_{F'BCG} + \nabla_{(B|F'}B_{A|CG)D} + \tfrac{1}{3}
    \phi_{(BC|D}{}^{F}H_{F'A|G)F} \nonumber \\& + 2
    \epsilon_{A(B}F_{|F'\vert CG)D} - \tfrac{1}{3}
    \phi_{(BC}{}^{FH}H_{|F'|G)FH}\epsilon_{AD} - \tfrac{1}{6}
    \phi_{(BC}{}^{FH}H_{|F'DFH}\epsilon_{A|G)} \nonumber\\ & -
    \tfrac{1}{3} \phi_{(B|D}{}^{FH}H_{F'|C|FH}\epsilon_{A|G)} -
    \tfrac{1}{6}
    \phi_{D(B}{}^{FH}H_{|F'|C|FH}\epsilon_{A|G)}\label{Eq:MiscIdentity2}
\end{align}
from which it follows that 
\begin{align}
    \kappa^{AD} \phi_{AD}{}^{GH} \nabla_{HA'}\phi_{BCFG} &= 4
    \xi^{A}{}_{A'}\phi_{(BC}{}^{DG}\phi_{F)ADG} + \tfrac{1}{2}
    \phi_{ADGH} \phi^{ADGH} H_{A'BCF} \nonumber\\ & - 4
    B_{(B}{}^{ADG}\nabla_{|AA'|}\phi_{CF)DG} - 8
    \phi_{(BC}{}^{AD}F_{|A'|F)ADG}\nonumber \\ & - 4
    \phi_{(B}{}^{ADG}\nabla_{|AA'|}B_{CF)DG} - \tfrac{1}{3}
    \phi_{(BC}{}^{AD}\phi_{|AD}{}^{GH}H_{A'|F)GH}\nonumber \\ & -
    \tfrac{2}{3}
    \phi_{(B}{}^{ADG}\phi_{C|AD}{}^{H}H_{A'|F)GH} \label{Eq:MiscIdentity3}
\end{align}
---see Appendix \ref{Appendix_A} for details.
\\

Now, commuting derivatives and using equation
\eqref{RescaledWeylEquationDisplayed} gives
\begin{align}
    \nabla^A{}_{A'}\nabla_{\bm\xi}\phi_{ABCD} &=
    \nabla_{\bm\xi}(\nabla^A{}_{A'}\phi_{ABCD}) +
    (\nabla^A{}_{A'}\xi^{FF'})\nabla_{FF'} \phi_{ABCD} +
    \xi^{FF'}\left[\nabla^A{}_{A'},
      \nabla_{FF'}\right]\phi_{ABCD}\nonumber\\ & =
    (\nabla^A{}_{A'}\xi^{FF'})\nabla_{FF'} \phi_{ABCD} +
    \xi^{FF'}\left[\nabla^A{}_{A'}, \nabla_{FF'}\right]\phi_{ABCD}.
\end{align}
Then, using equations \eqref{Eq:DecompGradXi},
\eqref{RescaledWeylEquationDisplayed}
and expanding the commutator gives
\begin{align*}
\nabla^D{}_{A'}\nabla_{\bm\xi}\phi_{ABCD} & = 6 \Lambda \xi^{D}{}_{A'} \phi_{ABCD} -  \xi^{DF'} \Phi_{D}{}^{F}{}_{A'F'} \phi_{ABCF} -  3\xi^{DF'} \Phi_{(A}{}^{F}{}_{\vert A'F'\vert} \phi_{BC)DF}  \nonumber\\
    &\quad -  3\Xi\xi^{D}{}_{A'}  \phi_{DFG(A} \phi_{BC)}{}^{FG} - 3\Lambda \kappa^{DF} \nabla_{FA'}\phi_{ABCD} + \tfrac{1}{4} Q^{DF} \nabla_{FA'}\phi_{ABCD}  \nonumber\\ 
    &\quad -  \tfrac{3}{2} \kappa^{DF} \Phi_{D}{}^{G}{}_{A'}{}^{F'} \nabla_{FF'}\phi_{ABCG} -  \tfrac{3}{2} \kappa^{DF} \Phi_{D}{}^{G}{}_{A'}{}^{F'} \nabla_{GF'}\phi_{ABCF} \nonumber\\
    &\quad + \tfrac{3}{4} \Xi \kappa^{DF} \phi_{DF}{}^{GH} \nabla_{HA'}\phi_{ABCG} - \tfrac{1}{2}( \nabla_{FF'}\phi_{ABCD}) \nabla_{Q(A'}H_{F')F}{}^{DQ}.
\end{align*}
Finally, using equations \eqref{Eq:MiscIdentity2}--\eqref{Eq:MiscIdentity3}, one obtains
\begin{align}
\nabla^D{}_{A'}\nabla_{\bm\xi}\phi_{ABCD}  &= -3 \Phi _{A}{}^{D}{}_{A'}{}^{F'} F_{F'BCD}  - \tfrac{3}{8} \Xi  \phi _{DFGH} \phi ^{DFGH} H_{A'ABC} + \Lambda  \phi _{BCDF} H_{A'A}{}^{DF} \nonumber\\ 
    & \quad- \tfrac{1}{8} Q^{DF} \nabla_{FA'}\phi _{ABCD}   - \tfrac{3}{2} \Phi ^{DF}{}_{A'}{}^{F'} \nabla_{FF'}B_{ABCD} + \tfrac{1}{4} (\nabla_{DA'}H^{F'DFG}) \nabla_{GF'}\phi _{ABCF} \nonumber\\
    &\quad + \tfrac{1}{4} (\nabla_{D}{}^{F'}H_{A'}{}^{DFG}) \nabla_{GF'}\phi _{ABCF}  + 12 \Lambda  F_{A'ABC} +6\Xi \nabla^G{}_{A'}\left( B_{(AB}{}^{DF}\phi_{CG)DF}\right) \nonumber \\
    &\quad - \tfrac{9}{2} \Phi _{(B}{}^{D}{}_{|A'}{}^{F'}F_{F'A|CD)}  - \tfrac{3}{2} \Phi ^{DF}{}_{A'}{}^{F'}\nabla_{(B|F'}B_{A|CD)F} + \tfrac{1}{4} \Xi  \phi _{(BC}{}^{GH}\phi _{DF)GH} H_{A'A}{}^{DF}  \nonumber\\
    &\quad+ 2 \Lambda  \phi_{DF(AB}H_{\vert A'\vert C)}{}^{DF}+ 3 \Xi  \phi _{(BC}{}^{DF}F_{|A'A|D)F} + 6 \Xi  \phi _{A(B}{}^{DF}F_{|A'|CD)F} \nonumber \\
    &\quad  + \tfrac{3}{8} \Phi _{(B}{}^{D}{}_{|A'|}{}^{F'}\phi _{CD)}{}^{FG}H_{F'AFG} + \tfrac{9}{8} \Phi _{(B}{}^{D}{}_{|A'}{}^{F'}\phi _{A|C}{}^{FG}H_{|F'|D)FG} \nonumber\\
    &\quad + \Phi _{A}{}^{D}{}_{A'}{}^{F'}\phi _{(BC}{}^{FG}H_{|F'|D)FG}  - \tfrac{1}{2} \Phi ^{DF}{}_{A'}{}^{F'}\phi _{A(BC}{}^{G}H_{|F'|D)FG} \nonumber \\
    &\quad  - \tfrac{1}{2} \Phi ^{DF}{}_{A'}{}^{F'}\phi _{A(B|D}{}^{G}H_{F'|CF)G}+ \tfrac{1}{6}\Xi H_{A'A}{}^{DF}\phi_{(BC}{}^{GH}\phi_{DF)GH} \nonumber\\
    &\quad + \tfrac{1}{6}\Xi H_{A'B}{}^{DF}\phi_{(CA}{}^{GH}\phi_{DF)GH} + \tfrac{1}{6}\Xi H_{A'C}{}^{DF}\phi_{(AB}{}^{GH}\phi_{DF)GH}. \label{Eq:CollineationIdentity}
\end{align}
Note that the final expression is homogeneous in the zero-quantities
$(\bmH, \bmB,\bmF)$ and their first derivatives, as required.\\

Collecting together the above, we derive the required wave equation for
$F_{A'BCD}$:
\begin{align}
    \square F_{A'BCD} &= \square (\nabla^A{}_{A'}B_{ABCD}) \nonumber\\
    &= \left[\square, \nabla^A{}_{A'}\right] B_{ABCD} + \nabla^A{}_{A'}\square B_{ABCD} \\
    &= \left[\square, \nabla^A{}_{A'}\right] B_{ABCD} + \tfrac{2}{3}\nabla^A{}_{A'}\nabla_{\bm\xi}\phi_{ABCD} \nonumber\\
    &\quad+ \nabla^A{}_{A'}\left(8\Lambda B_{ABCD} - 14\Xi \phi_{(AB}{}^{FG}B_{CD)FG} + \tfrac{2}{3}(\nabla_{FA'}\phi_{G(ABC})H^{A'}{}_{D)}{}^{FG}\right)\nonumber\\
    &=   \tfrac{2}{3}\nabla^A{}_{A'}\nabla_{\bm\xi}\phi_{ABCD} -6 \Lambda F_{A'BCD} - 6 \Phi_{(B}{}^{A}{}_{\vert A'}{}^{B'} F_{B'\vert CD)A} - 9 (\nabla^{A}{}_{A'}\Lambda)B_{BCDA} \nonumber \\
    &\quad + 3B_{(BC}{}^{FG} \phi_{D)AFG} \nabla^{A}{}_{A'}\Xi  + 3B_{AF(BC} \nabla^{FB'}\Phi_{D)}{}^{A}{}_{A'B'}  - 6 \Xi \phi_{(B}{}^{AFG} \nabla_{\vert GA'\vert}B_{CD)AF} \nonumber\\
    &\quad + \nabla^A{}_{A'}\left(8\Lambda B_{ABCD} - 14\Xi \phi_{(AB}{}^{FG}B_{CD)FG} + \tfrac{2}{3}(\nabla_{FA'}\phi_{G(ABC})H^{A'}{}_{D)}{}^{FG}\right) \label{Eq:WaveEqForF},
\end{align}
where we are using equation \eqref{Eq:SecondWaveEqForB} in the third line and
in the fourth we are expanding out the commutator and using the
Bianchi identites. Substituting equation
\eqref{Eq:CollineationIdentity} and setting $Q_{AB}=0$, we obtain a
homogeneous expression in $(\bmH, \bmB,\bmF)$ and their first
derivatives, as required.  \\

With this closed system of homogeneous
wave equations ---\eqref{Eq:WaveEqForH},
\eqref{Eq:FirstWaveEqForB} and \eqref{Eq:WaveEqForF}--- at hand, a direct application of Theorem
\ref{TheoremHomogeneousWave} gives the following:
\begin{proposition}\label{Prop:Propagation_KS}
  Given initial data for the conformal Einstein field equations on
  $\mathcal{U}\subset\mathcal{S}$ where $\mathcal{S}$ is a spacelike
  hypersurface $\mathcal{S}$ with normal vector $\tau^{AA'}$, then a
  (valence-2) Killing spinor candidate on $\mathcal{D}^{+}(\mathcal{U})$,
  is a true Killing spinor on an open neighbourhood $\mathcal{W}$
  of $\mathcal{U}$, with $\mathcal{W} \subseteq \mathcal{D}^{+}(\mathcal{U})$,
  if and only if
\begin{subequations}
\begin{eqnarray}
  &&
  H_{A'ABC}=B_{ABCD}=F_{A'ABC}=0,\label{Eq:KSInitialCondition1}\\ &&
  \nabla_{\bm\tau} H_{A'ABC}=\nabla_{\bm\tau} B_{ABCD}=
  \nabla_{\bm\tau} F_{A'ABC}=0, \label{Eq:KSInitialCondition2}
\end{eqnarray}
\end{subequations}
hold on $\mathcal{U}$.
\end{proposition}
\begin{proof}
The \emph{only if} direction is immediate. Suppose, on the other hand,
that $\kappa_{AB}$ is a Killing spinor candidate on
$\mathcal{D}^+(\mathcal{U})$ satisfying
\eqref{Eq:KSInitialCondition1}--\eqref{Eq:KSInitialCondition2} on
$\mathcal{U}$. In particular, $Q_{AB}=0$, and the identities
\eqref{Eq:WaveEqForH}, \eqref{Eq:FirstWaveEqForB} and
\eqref{Eq:WaveEqForF} reduce to a closed system of homogeneous wave
equations for the zero-quantities $\bmH,~\bmB$ and
$\bmF$. The uniqueness result for homogeneous wave equations,
given in Theorem \ref{TheoremHomogeneousWave},
ensures then that
\[ H_{A'ABC}=0,\qquad B_{ABCD}=0, \qquad F_{A'ABC}=0,\]
on an open neighbourhood $\mathcal{W}$ of $\mathcal{U}$, with $\mathcal{W}\subseteq
\mathcal{D}^{+}(\mathcal{U})$. In particular,
$\kappa_{AB}$ solves the Killing spinor equation on $\mathcal{W}$.
\end{proof}

To summarise: we have found necessary and sufficient conditions,
namely \eqref{Eq:KSInitialCondition1}--\eqref{Eq:KSInitialCondition2},
defined on $\mathcal{U}\subset \mathcal{S}$, for a given Killing
spinor candidate ---i.e. a field satisfying the wave equation
encoded by $Q_{AB}=0$, see \eqref{Eq:WaveForKS}--- to be a Killing spinor.

\subsection{Conformal Killing spinor initial data conditions (CKSID)}
\label{IntrinsicCKSID}

By analogy with the twistor case in section
\ref{Sec:IntrinsicTwistor}, in this section we aim to reduce
\eqref{Eq:KSInitialCondition1}--\eqref{Eq:KSInitialCondition2} to a
set of intrinsic conditions ---that is to say, conditions on
$\kappa_{AB}$ that are computable at the level of an initial data set
for the CFEs.  \\

This time, the initial value problem of interest is the following:
\begin{equation}\label{KillingSpinorIVP}
    \left\{
\begin{array}{ll}
	 \square \kappa_{AB} + 4 \Lambda\kappa_{AB} - \Xi
         \phi_{ABCD}\kappa^{CD} = 0 & \qquad
         \text{on}~\mathcal{D}^{+}(\mathcal{U}),\\ \kappa_{AB} =
         \bar{\kappa}_{AB}
         &\qquad\text{on}~\mathcal{U},\\ \nabla_{\bm\tau} \kappa_{AB}
         + \xi_{AB} = 0 &\qquad\text{on}~\mathcal{U}.
\end{array} \right.
\end{equation}
where $\xi_{AB}:=\mathcal{D}_{(A}{}^C\kappa_{B)C}$ is used as a
shorthand.  We define
\[\mathcal{H}_{ABCD}:=\tau_{(A}{}^{A'}H_{\vert A'\vert BCD)}, \qquad \mathcal{H}_{BC}:=\tau^{AA'}H_{A'ABC}, \]
in terms of which
\[H_{A'ABC} = -\tfrac{1}{2}\tau^D{}_{A'}\mathcal{H}_{ABCD} +
\tfrac{1}{8}\tau_{AA'}\mathcal{H}_{BC} + \tfrac{1}{8}\tau_{BA'}\mathcal{H}_{AC}
+ \tfrac{1}{8}\tau_{CA'}\mathcal{H}_{AB}. \]
Analogous to the twistor case, $\mathcal{H}_{AB}$ and
$\mathcal{H}_{ABCD}$ together contain all the information of
$H_{A'ABC}$; in particular,
\[
H_{A'ABC}=0 \quad \iff \quad \mathcal{H} _{AB}=0 \quad \& \quad
\mathcal{H}_{ABCD}=0.
\]
Substituting the definition of $H_{A'ABC}$ from
\eqref{KS_zero_quantities1},
\begin{equation}
\mathcal{H}_{BC}\equiv 3(\nabla_{\bm\tau} \kappa_{BC} + \xi_{BC}),
\qquad \mathcal{H}_{ABCD} \equiv
3\mathcal{D}_{(AB}\kappa_{CD)}. \label{Eq:SpatialKS}
\end{equation}
Note that $\mathcal{H}_{AB}|_{\mathcal{U}}=0$ is precisely the second
initial condition of \eqref{KillingSpinorIVP}. Again, the conditions
$\mathcal{H}_{ABCD}=B_{ABCD}=0$ involve only quantities intrinsic to
$\mathcal{S}$. Following the twistor case, we make the following
definition:

\begin{definition}{\em
A symmetric spinor field $\bar{\kappa}_{AB}$ defined on some
$\mathcal{U}\subset\mathcal{S}$ and satisfying
\begin{equation}
    \mathcal{H}(\bar{\bm\kappa})_{ABCD}\equiv
    3\mathcal{D}_{(AB}\bar{\kappa}_{CD)}=0,\qquad
    B(\bar{\bm\kappa})_{ABCD}\equiv
    \bar{\kappa}_{(A}{}^Q\phi_{BCD)Q}=0 \label{CSKIDs}
\end{equation}
will be called a \emph{conformal Killing spinor initial data set}
(\emph{CKSID}) on $\mathcal{U}$.}
\end{definition}

While conditions \eqref{CSKIDs} are clearly necessary for
$\kappa_{AB}$ to be a Killing spinor, we will only be able to prove
their sufficiency under an additional, albeit minor, assumption; see
the statement of Lemma \ref{prop_remove_redundant_conditions}, below.
We begin with the following:

\begin{lemma}\label{Lemma_initial_conditions}
  Suppose we have an initial data for the conformal Einstein field
  equations on $\mathcal{U}\subset\mathcal{S}$ where $\mathcal{S}$ is
  a spacelike hypersurface. Let $\bar{\kappa}_{AB}$ be a CKSID on
  $\mathcal{U}$ which moreover satisfies the equations
\begin{subequations}
\begin{eqnarray}
  && \bar{\kappa}_{(A}{}^{F}\mathcal{D}_{B}{}^{G}\phi_{CD)FG} +
  \phi_{(ABC}{}^{F}\bar{\xi}_{D)F} = 0,\label{RedundantCondition1}\\ &&
  \bar{\xi}^{FG}\mathcal{D}_{FG}\phi_{ABCD} +
  \tfrac{2}{3}\bar{\xi}\mathcal{D}_{(A}{}^F\phi_{BCD)F} =
  0, \label{RedundantCondition2}
\end{eqnarray}
\end{subequations}
where $\bar{\xi}_{AB}:=\mathcal{D}_{(A}{}^C\bar{\kappa}_{B)C}$ and
$\bar{\xi}:=\mathcal{D}^{AB}\bar{\kappa}_{AB}$. Then the resulting
solution $\kappa_{AB}$ of system \eqref{KillingSpinorIVP} satisfies
\[H(\bm\kappa)_{A'ABC}=B(\bm\kappa)_{ABCD}=F(\bm\kappa)_{A'BCD}=\nabla_{\bm\tau}H(\bm\kappa)_{A'ABC}
=\nabla_{\bm\tau}B(\bm\kappa)_{ABCD}=\nabla_{\bm\tau}F(\bm\kappa)_{A'ABC}=0 \]
on $\mathcal{U}$.
\end{lemma}
\begin{proof}
  First we note that for $\kappa_{AB}$ satisfying equation
  \eqref{KillingSpinorIVP}, the following holds
\begin{equation}
\xi_{AB'}|_{\mathcal{U}} = -\tfrac{3}{2}
\tau^B{}_{B'}\xi_{AB} +
\xi\tau_{AB'}, \label{DecompXiOnSUsingInitialCondition}
\end{equation}
with $\xi:=\mathcal{D}^{AB}\kappa_{AB}$ and where, recall,
$\xi_{AB}=\mathcal{D}_{(A}{}^C\kappa_{B)C}$. To see this, note that
\begin{align}
\xi_{AB'} = - \tau_{B}{}^{A'} \tau^{B}{}_{B'}\xi_{AA'}
&= - \tau_{B}{}^{A'} \tau^{B}{}_{B'}\nabla^C{}_{A'}\kappa_{AC}\nonumber\\
&=\tfrac{1}{2}\tau^{B}{}_{B'}\nabla_{\bm\tau}\kappa_{AB} +
\ \tau^{B}{}_{B'}
\mathcal{D}_{BC}\kappa_{A}{}^{C} \nonumber\\ &=\tfrac{1}{2}\tau^{B}{}_{B'}\nabla_{\bm\tau}\kappa_{AB}
- \tau^B{}_{B'}\xi_{AB} + \xi\tau_{AB'}\label{AuxXiExpanded} 
\end{align}
where we are decomposing the covariant derivative. Finally,
substituting the initial condition from \eqref{KillingSpinorIVP}, we
obtain equation \eqref{DecompXiOnSUsingInitialCondition}. Note also
that $\bar{\xi}=\xi$, $\bar{\xi}_{AB} = \xi_{AB}$ on $\mathcal{U}$ as
a result of the initial condition $\kappa_{AB}=\bar{\kappa}_{AB}$.
Starting from equation \eqref{IrrDecompCurlOfH}, performing the
decomposition of the covariant derivative and substituting the Killing
spinor candidate equation $Q_{AB}=0$, we get
\[ \nabla_{\bm\tau} H_{B'ABC} + 2 \tau^{D}{}_{B'} \tau^{FA'}
\mathcal{D}_{DF}H_{A'ABC} = -12 \Xi \tau^{D}{}_{B'}B_{ABCD}.\]
At this point we see that, since the CKSID conditions hold by assumption, then
\[\nabla_{\bm\tau} H_{B'ABC}\big\vert_{\mathcal{U}}=0\] also. Now, a
similar computation to the twistor case yields
\begin{equation}
    \nabla_{\bm\tau} B_{ABCD}|_{\mathcal{U}} = 2
    \kappa_{(A}{}^{F}\mathcal{D}_{B}{}^{G}\phi_{CD)FG} +
    \phi_{(ABC}{}^{F}\xi_{D)F},\label{EvolutionForBuchdahl}
\end{equation}
where we are again making use of equations
\eqref{RescaledWeyl_evo_const} and
\eqref{DecompXiOnSUsingInitialCondition}.  Note that the quantity on
the right-hand-side is intrinsic to $\mathcal{S}$. Hence, if we assume
\eqref{RedundantCondition1}, then we have
\begin{equation}\label{Eq:NormalDerivativeOfBuchdahlVanishes}
    \nabla_{\bm\tau}B_{ABCD} |_{\mathcal{U}}= 0.
\end{equation}
Consider now the quantity $F_{A'ABC}$. Recall that, by definition,
$F_{A'ABC}=\nabla^D{}_{A'}B_{ABCD}$, and so decomposing the covariant
derivative one obtains
\[F_{A'BCD} = \tfrac{1}{2} \tau^{A}{}_{A'} \nabla_{\bm\tau} B_{ABCD}
- \tau^{F}{}_{A'} \mathcal{D}^{A}{}_{F}B_{ABCD}.\] Hence, substituting
the CKSID condition $B_{ABCD}=0$ and using
\eqref{Eq:NormalDerivativeOfBuchdahlVanishes} we see that
$F_{A'ABC}=0$ on $\mathcal{U}$, also.  Thus it only remains to show
that $\nabla_{\bm\tau}F_{A'ABC}=0$ on $\mathcal{U}$.

\medskip
Comparing the two wave equations for $B_{ABCD}$,
\eqref{Eq:FirstWaveEqForB} and \eqref{Eq:SecondWaveEqForB}, we derive
the identity
\begin{multline}
     \nabla_{FA'}F^{A'}{}_{BCD} = \tfrac{1}{3}
     \nabla_{\xi}{}\phi_{BCDF} - 2 \Lambda B_{BCDF} + 3\Xi
     B_{(BC}{}^{AG} \phi_{D)FAG} - 7 \Xi B_{(BC}{}^{AG}\phi_{DF)AG}
     \\- \tfrac{1}{2} \phi_{(BCD}{}^{A}Q_{F)A} + \tfrac{1}{3}
     H^{A'}{}_{(B}{}^{AG}\nabla_{|AA'|}\phi_{CDF)G}\label{CurlOfFInTermsOfCollineation}.
\end{multline}
Decomposing the covariant derivative and imposing the Killing spinor
candidate equation, $Q_{AB}=0$, we get
\begin{multline}
\nabla_{\bm\tau} F_{B'BCD} + 2 \tau^{A}{}_{B'} \tau^{FA'}
\mathcal{D}_{AF}F_{A'BCD} \\ = \tau^{A}{}_{B'} \big(\tfrac{2}{3}
\nabla_{\xi}\phi_{BCDA} - 4 \Lambda B_{BCDA} + 6 \Xi B_{(BC}{}^{FG}
\phi_{D)AFG} \\ - 14 \Xi B_{(BC}{}^{AG}\phi_{DA)FG} + \tfrac{2}{3}
H^{A'}{}_{(B}{}^{FG}\nabla_{|FA'|}\phi_{CDA)G} \big).
\end{multline}
Hence, if
$F_{A'ABC}|_{\mathcal{U}}=H_{A'ABC}|_{\mathcal{U}}=B_{ABCD}|_{\mathcal{U}}=0$,
then $\nabla_{\bm\tau}F_{A'ABC}|_{\mathcal{U}}=0$ if and only if
\[ \nabla_{\bm\xi}\phi_{ABCD}|_{\mathcal{U}}=0.\]
Decomposing the covariant derivative, and using the evolution equation
\eqref{RescaledWeyl_evo_const} for $\phi_{ABCD}$ again, along with
equation \eqref{DecompXiOnSUsingInitialCondition}, we have
\[ \nabla_{\bm\xi}\phi_{ABCD}|_{\mathcal{U}} =  \xi \mathcal{D}_{(A}{}^{F}\phi_{BCD)F} + \tfrac{3}{2} \xi^{FG} \mathcal{D}_{FG}\phi_{ABCD}. \]
Hence, if we assume \eqref{RedundantCondition2} to hold, then
\[\nabla_{\bm\tau} F_{B'BCD}|_{\mathcal{U}}=0, \]
and the result of Lemma \ref{Lemma_initial_conditions} follows.
\end{proof}
Condition \eqref{RedundantCondition1} is in fact the ``unphysical"
counterpart of the condition appearing in \cite{GarVal08c}, which was
later shown to be redundant in \cite{BaeVal10c}, modulo a minor
algebraic assumption on the Killing spinor initial data ---see (i) and
(ii) in the lemma below. Although no such counterpart of
\eqref{RedundantCondition2} appears in the physical case, this same
algebraic assumption ensures redundancy of both
\eqref{RedundantCondition1} and \eqref{RedundantCondition2}:
\begin{lemma}\label{prop_remove_redundant_conditions}
Suppose that $\bar{\kappa}_{AB}$ is CKSID set on $\mathcal{U}$
satisfying one of the following two conditions on $\mathcal{U}$:
\begin{enumerate}[(i)]
    \item $\bar{\kappa}_{AB}\bar{\kappa}^{AB} \equiv 0$ ~~~but
      ~~~$\bar{\kappa}_{AB}\neq 0$,~~~\textbf{\textit{or}}
    \item $\bar{\kappa}_{AB}\bar{\kappa}^{AB}\neq 0$.
\end{enumerate}
Then
conditions \eqref{RedundantCondition1} and \eqref{RedundantCondition2}
are redundant in the sense that they are automatically satisfied by
$\bar{\kappa}_{AB}$ by virtue of the CKSID conditions \eqref{CSKIDs}. 
\end{lemma}
The proof of this lemma requires decomposing the fields respect to a spin dyad
and considering the cases where $\bm\phi$ is of different Petrov types.
This is a long but direct calculation that is given in Appendix \ref{Sec:ProofOfProp3}.
\begin{remark}{\em 
Note that if it is not the case that $\bar{\kappa}_{AB}\bar{\kappa}^{AB}\equiv 0$ on
$\mathcal{U}$, then there must exist an open subset
$\emptyset\neq \tilde{\mathcal{U}}\subset\mathcal{U}$ on which
(ii) holds, since the vanishing of
$\bar{\kappa}_{AB}\bar{\kappa}^{AB}$ is a closed condition by the assumed
continuity (in fact, differentiability) of the Killing spinor
candidate. Hence, the conditions of Lemma \eqref{prop_remove_redundant_conditions} imply, at worst, a restriction the of domain of applicability.}
\end{remark}

\noindent
Finally, putting together Proposition \ref{Prop:Propagation_KS},
Lemmas \ref{Lemma_initial_conditions} and
\ref{prop_remove_redundant_conditions} gives the following valence-2
analogue of Theorem \ref{Theorem_twistor}:

\begin{theorem}\label{Theorem_KS}
Consider an initial data set for the vacuum conformal Einstein field
equations, as encoded in the CFE zero-quantities
\eqref{Def_ConfFactor_CFE_zeroquant}--\eqref{Def_cons_CFE_zeroquant},
on a spacelike hypersurface $\mathcal{S}$ and let
$\mathcal{U}\subset\mathcal{S}$ be an open set.  If there exists a
CKSID set $\bar{\kappa}_{AB}$ on $\mathcal{U}$ satisfying either of
conditions $(i)$ or $(ii)$ from Lemma
\ref{prop_remove_redundant_conditions}, then the development of the
initial data set admits a Killing spinor on an open neighbourhood 
$\mathcal{W}$ of $\mathcal{U}$, with $\mathcal{W}\subseteq\mathcal{D}^+(\mathcal{U})$,
given by the solution $\kappa_{AB}$ of
the initial value problem \eqref{KillingSpinorIVP}. If condition (i) holds, then in fact
there exists a twistor on $\mathcal{W}$.
\end{theorem}

\begin{proof}

Given such a $\bar{\kappa}_{AB}$, Lemmas
\ref{Lemma_initial_conditions} and
\ref{prop_remove_redundant_conditions} together imply that the Killing
spinor candidate $\kappa_{AB}$ constructed as a solution of
\eqref{KillingSpinorIVP} satisfies conditions
\eqref{Eq:KSInitialCondition1}--\eqref{Eq:KSInitialCondition2}. Proposition
\ref{Prop:Propagation_KS} then implies that $H_{A'ABC}=0$ on
an open neighbourhood $\mathcal{W}$ of $\mathcal{U}$, with
$\mathcal{W} \subseteq \mathcal{D}^+(\mathcal{U})$,
and hence that $\kappa_{AB}$ is indeed a
Killing spinor on $\mathcal{W}$. In particular, if condition (i) holds then
$\bar{\kappa}_{AB}=\bar{\kappa}_A\bar{\kappa}_B$ for some
$\bar{\kappa}_A\neq 0$ on $\mathcal{U}$. It is straightforward to
verify that $\bar{\kappa}_A$ solves the conformal twistor initial data
conditions, equations \eqref{twistor_CSKIDs}; see the discussion of
Appendix \ref{TypeNCase}. Theorem \ref{Theorem_twistor} then implies
the existence of a twistor on $\mathcal{W}$.
\end{proof}

\begin{remark}{\em 
Note that we recover the conditions from \cite{BaeVal10a}, namely
\[\tilde{\mathcal{D}}_{(AB}\tilde{\kappa}_{CD)}=0, \qquad \tilde{\kappa}_{(A}{}^F
\Psi_{BCD)F}=0, \] when $\Xi\equiv 1$ on $\mathcal{U}$.  }
\end{remark}
Given the close connection between the notion of algebraic special Petrov types
and the existence of Killing spinors, discussed in section
\ref{Background}, it is not surprising that one can use the previous
result to establish conditions under which the Petrov type of an
initial data set is ``propagated" to the resulting spacetime
development; see Theorem 3 of \cite{GarVal08c} for a similar result in
the physical framework case. This is the content of the following:

\begin{corollary}\label{Corollary:PetrovPropagation}
  Given initial data for the CFEs, suppose that the initial data for the rescaled Weyl spinor
  $\phi_{ABCD}$ is of Petrov type D on $\mathcal{U}\subset\mathcal{S}$
  and suppose further that
  \footnote{The fields $\sigma$ and $\lambda$ (not to be confused with
  the cosmological constant), along with $\kappa, \tau, \rho, \mu,
  \nu, \pi$ to follow, are Newmann--Penrose (NP) scalars
  \cite{PenRin84, PenRin86}; see the Appendix for further details.}
\begin{equation} 
\sigma\equiv - o^A o^B o^C\mathcal{D}_{AB}o_C=0 \quad \text{and}\quad
\lambda\equiv
\iota^A\iota^B\iota^C\mathcal{D}_{AB}\iota_C=0\label{ShearConditions}
\end{equation}
hold on $\mathcal{U}$, where $\lbrace \bmo, \bm\iota\rbrace$ is an
adapted spin dyad in terms of which $\phi_{ABCD}=\phi
o_{(A}o_B\iota_C\iota_{D)}$ for some field $\phi:\mathcal{U}\rightarrow\mathbb{C}$.  Then, the rescaled
Weyl spinor $\phi_{ABCD}$ of the corresponding development $(\mathcal{M},\bmg)$ 
is of Petrov type D on an an open neighbourhood $\mathcal{W}$ of
$\mathcal{U}$, with $\mathcal{W}\subseteq\mathcal{D}^+(\mathcal{U})$, as evidenced by the
existence of a non-trivial (valence-2) Killing spinor satisfying
$\kappa_{AB}\kappa^{AB}\neq 0$.
\end{corollary}
\begin{proof}
  Suppose that the initial data for the rescaled Weyl spinor
  $\phi_{ABCD}$ is of Petrov type D on $\mathcal{U}$, with
$\lbrace \bmo, \bm\iota\rbrace$ an adapted spin dyad in terms of which
$\phi_{ABCD}=\phi o_{(A}o_B\iota_C\iota_{D)}$ for some field
$\phi$ on $\mathcal{U}$. Note that $\phi\neq 0$ everywhere on $\mathcal{U}$; if
$\phi=0$ at some $p\in\mathcal{U}$ then $\phi_{ABCD}|_{\mathcal{U}}$ would be of type O
at $p$. Define $\bar{\kappa}_{AB}=\phi^{-1/3}o_{(A}\iota_{B)}$ and
note that $\bar{\kappa}_{AB}\bar{\kappa}^{AB}\propto\phi^{-2/3}\neq 0$
on $\mathcal{U}$. A short calculation verifies that the Buchdahl
constraint $B(\bar{\bm\kappa})_{ABCD}=0$ is satisfied. The equation
$\mathcal{D}_{(AB}\bar{\kappa}_{CD)}=0$, on the other hand, decomposes
to give
\begin{eqnarray*}
&& \sigma=0,\\
&& \mathcal{D}_{\bm0\bm0}\phi = 3\phi(\kappa - \tau), \\
&& \mathcal{D}_{\bm0\bm1}\phi = \tfrac{3}{2} \phi(\rho + \mu),\\
&& \mathcal{D}_{\bm1\bm1}\phi = 3\phi(\nu - \pi), \\
&& \lambda = 0,
\end{eqnarray*}
in terms of the NP scalars ---to see this, substitute
$e^{\varkappa}=\phi^{-1/3}$ in equations
\eqref{SpatialSenTypeD0000}--\eqref{SpatialSenTypeD1111} of the
Appendix.  The first and last components are guaranteed by the
assumption \eqref{ShearConditions}.  The remaining equations are
precisely those implied by the constraint
$\mathcal{D}^{AB}\phi_{ABCD}=0$ ---see \eqref{WeylConstraintTypeD} of
the Appendix.  Hence, the conformal Killing spinor initial data
conditions \eqref{CSKIDs} are met by $\bar{\kappa}_{AB}$ and
propagating this data using \eqref{KillingSpinorIVP} we obtain a Killing spinor
$\kappa_{AB}$ on $\mathcal{D}^+(\mathcal{U})$, according to Theorem
\ref{Theorem_KS}. Now, by continuous dependence on the data, it
follows that there exists some open subset
$\mathcal{V}\subset\mathcal{D}^+(\mathcal{U})$ on which
$\kappa_{AB}\kappa^{AB}\neq 0$. Since the Buchdahl constraint
neceessarily holds on $\mathcal{V}$ and since
$\kappa_{AB}\kappa^{AB}\neq 0$, it follows that either $\phi_{ABCD}$ is of
Petrov type D or O on $\mathcal{V}$. Now
$\phi_{ABCD}\phi^{ABCD}\propto \phi^2\neq 0$ on $\mathcal{U}$ and so
by continuity $\phi_{ABCD}\phi^{ABCD}\neq 0$ on some open
$\tilde{\mathcal{V}}\subset\mathcal{V}$. Therefore, $\phi_{ABCD}$ cannot
be of type O anywhere on $\tilde{\mathcal{V}}$ and so must be of type
D.
\end{proof}

\begin{remark}{\em 
 While the constraint $\mathcal{D}^{AB}\phi_{ABCD}=0$ implies that the
 $\bm0\bm0\bm0\bm1$, $\bm0\bm0\bm1\bm1$ and
 $\bm0\bm1\bm1\bm1$-components of
 $\mathcal{D}_{(AB}\bar{\kappa}_{CD)}=0$ are all satisfied, as shown
 in the above proof, the $\bm0\bm0\bm0\bm0$ and
 $\bm1\bm1\bm1\bm1$-spin components are not guaranteed simply as a
 consequence of the expression
 $\bar{\kappa}_{AB}=\phi^{-1/3}o_{(A}\iota_{B)}$. Hence, the
 conditions of Theorem \ref{Theorem_KS} are more restrictive than
 simply assuming that $\phi_{ABCD}$ be of Petrov type D on
 $\mathcal{U}$. These remaining two components of
 $\mathcal{D}_{(AB}\bar{\kappa}_{CD)}=0$, namely $\sigma=\lambda=0$ are however implied by the
 $\bm0\bm0\bm0\bm0$ and $\bm1\bm1\bm1\bm1$-components of the
 \emph{evolution equation}
 $\nabla_{\bm\tau}\phi_{ABCD}=2\mathcal{D}_{(A}{}^F\phi_{BCD)F}$ if
 one \emph{assumes} that the Petrov type extends to the spacetime
 development, consistent with Remark
 \ref{Remark:DyadExpressionForKillingSpinorInTypeD}.}
\end{remark}

\section*{Conclusions}

In this article a \emph{conformal} version of the Killing spinor
initial data equations given in \cite{GarVal08c}, namely equations \eqref{CSKIDs}, are derived. By
conformal it is understood that $(\mathcal{M},\bmg)$ is conformally
related to a vacuum Einstein spacetime $(\tilde{\mathcal{M}},\tilde{\bmg})$.
It is shown that, modulo a minor technical assumption, the existence of a non-trivial solution of equations
\eqref{CSKIDs} is a necessary and sufficient
condition for the existence of a Killing spinor on the development in
the unphysical spacetime $(\mathcal{M},\bmg)$, as given by a solution of Friedrich's
\emph{conformal Einstein field equations}.
  The initial data equations \eqref{CSKIDs} are comprised of one differential condition and one
  algebraic condition.  The differential condition corresponds to the
  so-called \emph{spatial Killing spinor equation} while the algebraic
  condition is a restriction imposed by the \emph{Buchdahl constraint} on the initial
  hypersurface. This constraint can be interpreted as restricting the Petrov type
  of the initial data set for the conformal Einstein field equations.
  Although conditions \eqref{CSKIDs} look formally
  identical to those derived for the physical spacetime
  $(\tilde{\mathcal{M}},\tilde{\bmg})$, only with the Weyl spinor being replaced by the rescaled Weyl spinor, the derivation of these
  conditions in an unphysical spacetime
  $(\mathcal{M},\bmg)$ is non-trivial, owing to the non-trivial behaviour of the Einstein field equations under conformal conformal rescaling.  
  \medskip

In the case where the conformal rescaling is
trivial i.e. $\Xi \equiv 1$, we recover the results of
\cite{GarVal08c}. However, even for this case the set of variables to
be propagated in the physical and the unphysical spacetimes
are different. This difference can be traced back to the observation
that for a general Lorentzian manifold the vector
$\xi_{AA'}=\nabla^{B}{}_{A'}\kappa_{AB}$ is not a Killing vector.
Furthermore, even in the case where $(\mathcal{M},\bmg)$ satisfies the
conformal Einstein field equations, this vector does not correspond to
a conformal Killing vector as one could naively expect but rather to a
collineation for the rescaled Weyl tensor/spinor, as shown here. Naturally, once the existence of a
Killing spinor in the unphysical spacetime $(\mathcal{M},\bmg)$ is
established one can always construct, a posteriori, a conformal
Killing vector $X_{AA'}$ on $(\mathcal{M},\bmg)$ which corresponds to a Killing vector
$\tilde{\xi}_{AA'}$ of the physical spacetime
$(\tilde{\mathcal{M}},\tilde{\bmg})$.
Notice that the conformal approach followed in this article i.e.
the use of the conformal Einstein field equations, allows for the
possibility of $\mathcal{S}$ intersecting non-trivially with (or even
being a subset of) null infinity $\mathscr{I}$.  One possible
application is the characterisation of \emph{asymptotic initial data}
---initial data for the conformal Einstein field equations given on
$\mathscr{I}$--- for de Sitter-like spacetimes and, in particular, an
asymptotic characterisation of the Kerr-de Sitter spacetime, via the
existence of Killing spinors at the (spacelike) conformal
boundary. Such a result would mirror the spinorial characterisations
of Kerr given in \cite{BaeVal10a,BaeVal10b,BaeVal10c,BaeVal11b}. It
would also be of interest to compare with \cite{Gar16, Gar20}, in which
the propagation of Petrov type is explored in the context of both the
Einstein and the \emph{conformal} Einstein field equations.
\medskip

The applications of the core analysis of this article, however, are
not restricted to the study of de Sitter-like spacetimes. Indeed, the
most taxing part of the procedure, generally speaking, consists of finding a closed system
of homogeneous wave equations for the relevant zero-quantities, these equations being irrespective of the causal nature
of $\mathcal{S}$. For the spacelike case, the uniqueness result for solutions of homogeneous wave equations given by Theorem \ref{TheoremHomogeneousWave} ensures that if
trivial initial data on $\mathcal{S}$ is provided then $H_{A'ABC}$,
$B_{ABCD}$ and $F_{A'ABC}$ vanish on the domain of dependence of the
data.  Analogous theorems for the characteristic problem or the
initial boundary value problem could be used to obtain similar
conditions on a null or timelike hypersurface, as has been done in the
case of the Killing vector initial data equations and the physical
Killing spinor initial data equations ---see \cite{Pae14a,
  ColRacVal18} and \cite{CarVal18}.  In the case of the characteristic
problem, the conformal approach of this article would facilitate the
analysis of the existence of Killing spinors at the conformal boundary
of an asymptotically flat spacetime.  In the case of a timelike
hypersurface, the analogous conformal Killing spinor initial data
equations could potentially be used in the analysis and characterisation of
anti-de Sitter like spacetimes.

\subsection*{Acknowledgements}

We would like to thank Juan A. Valiente Kroon, David Hilditch and
Justin Feng for helpful discussions.  E. Gasper\'in acknowledges
support from Consejo Nacional de Ciencia y Tecnolog\'ia (Mexico)
---CONACyT studentship 494039/218141--- in the early stages of this
work and from Fundaç\~ao para a Ci\^encia e a Tecnologia (Portugal)
---FCT-2020.03845.CEECIND--- during its completion. J. L. Williams
acknowledges support from the COST grant CA16104.  \appendix

\section{Appendix}\label{Appendix_A}

Here we give the derivation of some of the key identities used in
section \ref{Sec:KSWaveEqs} and the proof of Lemma
\ref{prop_remove_redundant_conditions} from section
\ref{IntrinsicCKSID}.

\subsection{Spinorial identities}

Here we give the derivation of identities \eqref{Eq:MiscIdentity2} and
\eqref{Eq:MiscIdentity3}. First, by expanding out the definition of
$B_{ABCD}$,
\begin{align*}
\nabla^D{}_{G'}B_{ABFD} &= - \tfrac{1}{4} \phi_{ABFD}
\nabla_{CG'}\kappa^{CD} - \tfrac{3}{4} \kappa_{(A}{}^{C} \nabla_{\vert
  DG'\vert}\phi_{BF)C}{}^{D} \\ &\qquad - \tfrac{1}{4} \kappa^{CD}
\nabla_{DG'}\phi_{ABFC} + \tfrac{3}{4} \phi_{CD(AB}
\nabla^{D}{}_{\vert G'\vert}\kappa_{F)}{}^{C} \\ &=\tfrac{1}{2}
\xi^{C}{}_{G'} \phi_{ABFC} + \tfrac{1}{4} \phi_{CD(AB} H_{\vert
  G'\vert F)}{}^{CD} - \tfrac{1}{4} \kappa^{CD}
\nabla_{DG'}\phi_{ABFC}
\end{align*}
the second equality following from
$\nabla^A{}_{A'}\phi_{ABCD}=0$. Using \eqref{Eq:DecompGradKS}, the
definition of $F_{A'ABC}$ and rearranging,
\begin{equation}
 \kappa^{CD} \nabla_{DG'}\phi_{ABFC} = 2 \xi^{C}{}_{G'} \phi_{ABFC} +
 \phi_{CD(AB} H_{\vert G'\vert F)}{}^{CD} - 4
 F_{G'ABF}\label{Eq:MiscIdentity}
 \end{equation}
Then, using the irreducible decomposition
\begin{equation*}
\kappa_{F}{}^{D} \phi_{ABCD} = B_{ABCF} + \tfrac{1}{4} \kappa^{DG}
\phi_{BCDG} \epsilon_{AF} + \tfrac{1}{4} \kappa^{DG} \phi_{ACDG}
\epsilon_{BF} + \tfrac{1}{4} \kappa^{DG} \phi_{ABDG} \epsilon_{CF},
\end{equation*}
one calculates
\begin{align*}
	\kappa_{A}{}^{F} \nabla_{GF'}\phi_{BCDF} &=
        \nabla_{GF'}(\kappa_{A}{}^{F}\phi_{BCDF}) -
        \phi_{BCDF}\nabla_{GF'}\kappa_{A}{}^{F}\\ &=
        \nabla_{GF'}B_{ABCD} - \tfrac{3}{4} \kappa^{FH} \epsilon_{A(B}
        \nabla_{\vert GF'\vert}\phi_{CD)FH}\\ &\qquad - \phi_{BCDF}
        \nabla_{GF'}\kappa_{A}{}^{F} - \tfrac{3}{4}
        \epsilon_{A(B}\phi_{CD)FH}
        \nabla_{GF'}\kappa^{FH}\\ &=\nabla_{GF'}B_{ABCD} +
        \tfrac{1}{3} \xi_{AF'} \phi_{BCDG} - \tfrac{1}{3} \phi_{BCDF}
        H_{F'AG}{}^{F} + \tfrac{1}{2} \xi^{F}{}_{F'}\epsilon_{A(B}
        \phi_{CD)GF} \\ &\qquad - \tfrac{1}{4}
        \epsilon_{A(B}\phi_{CD)FH} H_{F'G}{}^{FH} - \tfrac{1}{3}
        \xi^{F}{}_{F'} \phi_{BCDF} \epsilon_{AG} - \tfrac{3}{4}
        \kappa^{FH} \epsilon_{A(B} \nabla_{\vert
          GF'\vert}\phi_{CD)FH},
\end{align*}
where the third equality follows from \eqref{Eq:DecompGradKS}. Then
swapping indices on the $\nabla\bm\phi$ terms and using
$\nabla^A{}_{A'}\phi_{ABCD}=0$, one obtains \eqref{Eq:MiscIdentity2}:
\begin{align*}
    \kappa_{A}{}^{F} \nabla_{FF'}\phi_{BCDG} &= \tfrac{1}{3} \xi_{AF'}
    \phi_{BCDG} - \tfrac{1}{3} \xi^{F}{}_{F'} \phi_{BCGF}
    \epsilon_{AD} - \tfrac{1}{3}
    \xi^{F}{}_{F'}\phi_{(BC|DF}\epsilon_{A|G)} \nonumber \\ & -
    \tfrac{2}{3} \xi^{F}{}_{F'}\phi_{(B|D|C|F}\epsilon_{A|G)} +
    \epsilon_{AD} F_{F'BCG} + \nabla_{(B|F'}B_{A|CG)D} + \tfrac{1}{3}
    \phi_{(BC|D}{}^{F}H_{F'A|G)F} \nonumber \\& + 2
    \epsilon_{A(B}F_{|F'\vert CG)D} - \tfrac{1}{3}
    \phi_{(BC}{}^{FH}H_{|F'|G)FH}\epsilon_{AD} - \tfrac{1}{6}
    \phi_{(BC}{}^{FH}H_{|F'DFH}\epsilon_{A|G)} \nonumber\\ & -
    \tfrac{1}{3} \phi_{(B|D}{}^{FH}H_{F'|C|FH}\epsilon_{A|G)} -
    \tfrac{1}{6}
    \phi_{D(B}{}^{FH}H_{|F'|C|FH}\epsilon_{A|G)}.
\end{align*}
To show \eqref{Eq:MiscIdentity3}, we first apply $\nabla^A{}_{A'}$ to
identity \eqref{Eq:UsefulIdentity2}, using \eqref{Eq:DecompGradKS} and
$\nabla^A{}_{A'}\phi_{ABCD}=0$, one obtains
\begin{align*}
&\tfrac{1}{2} \xi^{A}{}_{A'} \phi_{(BF}{}^{DG} \phi_{C)ADG} + \tfrac{1}{12} \
\phi_{ADGH} \phi_{BCF}{}^{H} H_{A'}{}^{ADG} + \tfrac{1}{4} \
\phi_{A(BF}{}^{H} \phi_{C)DGH} H_{A'}{}^{ADG}  \\
&-  \tfrac{3}{4} \
\kappa^{AD} \phi_{(CF}{}^{GH} \nabla_{\vert HA'\vert}\phi_{B)ADG} + \tfrac{1}{4} \
\kappa^{AD} \phi_{AD}{}^{GH} \nabla_{HA'}\phi_{BCFG} = 2 \nabla^G{}_{A'}(B_{(BC}{}^{AD} \phi_{FG)}{}_{AD})	
\end{align*}
Then, using \eqref{Eq:MiscIdentity2}, along with the irreducible decomposition 
\begin{equation*}
\phi_{ABCD} \phi_{FG}{}^{CD} = \tfrac{1}{6} \phi_{CDHL} \phi^{CDHL}
\epsilon_{AG} \epsilon_{BF} + \tfrac{1}{6} \phi_{CDHL} \phi^{CDHL}
\epsilon_{AF} \epsilon_{BG} +
\phi_{(AB}{}^{CD}\phi_{FG)CD}, 
\end{equation*}
one finally obtains \eqref{Eq:MiscIdentity3}: 
\begin{align*}
    \kappa^{AD} \phi_{AD}{}^{GH} \nabla_{HA'}\phi_{BCFG} &= 4
    \xi^{A}{}_{A'}\phi_{(BC}{}^{DG}\phi_{F)ADG} + \tfrac{1}{2}
    \phi_{ADGH} \phi^{ADGH} H_{A'BCF} \nonumber\\ & - 4
    B_{(B}{}^{ADG}\nabla_{|AA'|}\phi_{CF)DG} - 8
    \phi_{(BC}{}^{AD}F_{|A'|F)ADG}\nonumber \\ & - 4
    \phi_{(B}{}^{ADG}\nabla_{|AA'|}B_{CF)DG} - \tfrac{1}{3}
    \phi_{(BC}{}^{AD}\phi_{|AD}{}^{GH}H_{A'|F)GH}\nonumber \\ & -
    \tfrac{2}{3}
    \phi_{(B}{}^{ADG}\phi_{C|AD}{}^{H}H_{A'|F)GH}. \label{Eq:MiscIdentity3}
\end{align*}

\subsection{Proof of Lemma \ref{prop_remove_redundant_conditions}}
\label{Sec:ProofOfProp3}

Assuming $\bar{\kappa}_{AB}\neq 0$, the Buchdahl constraint restricts
the Petrov type of $\bm\phi$ to be type D, N or O.  We follow the same
strategy as in \cite{BaeVal10c}, expanding out conditions
\eqref{RedundantCondition1}--\eqref{RedundantCondition2} in an adapted
spin dyad, considering separately the cases (i):
$\bar{\kappa}_{AB}\bar{\kappa}^{AB}\equiv 0$, $\bar{\kappa}_{AB}\neq
0$ and (ii): $\bar{\kappa}_{AB}\bar{\kappa}^{AB}\neq 0$ on
$\mathcal{U}$, from Lemma \ref{prop_remove_redundant_conditions},
corresponding to Petrov types N and D, respectively.  Observe that for
Type O, for which $\phi_{ABCD}=0$, the proof of Lemma
\ref{prop_remove_redundant_conditions} trivialises, so only the types
$N$ and $D$ are needed.  \\

Recalling that $\mathcal{D}_{AB}:= \tau_{(A}{}^{A'}\nabla_{B)A'}$, a straightforward computation yields 
\begin{align*}
& o^Ao^Bo^C\mathcal{D}_{AB}o_C = -\sigma,\\
& o^Ao^B\iota^C\mathcal{D}_{AB}o_C=o^Ao^Bo^C\mathcal{D}_{AB}\iota_C=-\beta,\\
& o^A\iota^B o^C\mathcal{D}_{AB}o_C = \tfrac{1}{2}(\kappa - \tau),\\
& o^A\iota^B\iota^C\mathcal{D}_{AB}o_C = o^A\iota^B o^C\mathcal{D}_{AB}\iota_C = \tfrac{1}{2}(\epsilon - \gamma),\\
& \iota^A\iota^B o^C\mathcal{D}_{AB}o_C = \rho,\\
& \iota^A\iota^B\iota^C\mathcal{D}_{AB}o_C = \iota^A\iota^B o^C\mathcal{D}_{AB}\iota_C = \alpha,\\
& o^A o^B\iota^C\mathcal{D}_{AB}\iota_C = -\mu,\\
& o^A\iota^B\iota^C\mathcal{D}_{AB}\iota_C = \tfrac{1}{2}(\pi - \nu),\\
& \iota^A\iota^B\iota^C\mathcal{D}_{AB}\iota_C = \lambda	,
\end{align*}
where we are following the conventions of \cite{PenRin84} in the
definition of the Newmann--Penrose (NP) scalars $\alpha, \beta,
\epsilon, \gamma, \kappa, \mu, \lambda, \rho, \tau, \sigma, \nu, \pi.$
\\

The following identities will also be useful
\begin{subequations}
\begin{eqnarray}
&& \mathcal{D}_{AB}o^B = o^Bo^C\mathcal{D}_{AC}\iota_B - \iota^B o^C\mathcal{D}_{AB}o_C,\\
&& \mathcal{D}_{AB}\iota^B = \iota^B o^C(\mathcal{D}_{AC}\iota_B - \mathcal{D}_{AB}\iota_C),
\end{eqnarray}
\end{subequations}
and follow easily from $\epsilon_{AB}=o_A\iota_B - o_B\iota_A$. 

\subsubsection{Case I: $\bar{\kappa}_{AB}\bar{\kappa}^{AB}\equiv 0$, $\bar{\kappa}_{AB}\neq 0$ on $\mathcal{U}$} \label{TypeNCase}

The assumption on $\bar{\kappa}_{AB}$ imply that there exists a spin
dyad $\lbrace \bm\omicron, \bm\iota\rbrace$ on $\mathcal{U}$ such that
$\bar{\kappa}_{AB}=\omicron_A\omicron_B$. The Buchdahl constraint then
implies that
\begin{equation}
\phi_{ABCD}=\phi \omicron_A\omicron_B\omicron_C\omicron_D
\end{equation}
for some scalar field $\phi:\mathcal{U}\rightarrow\mathbb{C}$ and
hence that the curvature is of Petrov type N on $\mathcal{U}$. Note
that $\phi_{ABCD}\omicron^D=0$. The equation
$\mathcal{D}_{(AB}\bar{\kappa}_{CD)}=0$ implies
\begin{equation*}
 \omicron^A \omicron^B \omicron^C \mathcal{D}_{AB}\omicron_C=
 \omicron^A \omicron^B \iota^C\mathcal{D}_{(AB}\omicron_{C)}=
 \omicron^A \iota^B \iota^C\mathcal{D}_{(AB}\omicron_{C)}= \iota^A
 \iota^B \iota^C\mathcal{D}_{AB}\omicron_{C}=0,
\end{equation*}
implying that $\mathcal{D}_{(AB}\omicron_{C)}=0$ ---that is to say
that $\omicron_A$ is a twistor candidate. In terms of the NP scalars,
the above read as follows
\begin{equation}\label{NPRelationsTypeN} 
\sigma=- \beta + \kappa - \tau=\epsilon - \gamma + \rho=\alpha=0.
\end{equation}	
Using these relations, we obtain
\begin{equation}\label{AuxSpinorsTypeN}
 \bar{\xi} = -3\beta,\qquad \bar{\xi}_{AB} = 2\rho o_Ao_B - 2\beta
 o_{(A}\iota_{B)}.
\end{equation}
The non-trivial component of the constraint
$\mathcal{D}^{AB}\phi_{ABCD}=0$ reduces to
\begin{equation}
 \mathcal{D}_{\bm0\bm0}\phi = \tfrac{5}{3}\phi (2\beta + \kappa -
 \tau)=5\phi\beta.\label{TypeNweylConstraint}
\end{equation}
Then, substituting \eqref{AuxSpinorsTypeN}, condition
\eqref{RedundantCondition1} reduces to
\begin{align*}
\omicron^A \iota^B \iota^C \iota^D
\left(2\bar{\kappa}_{(A}{}^{F}\mathcal{D}_{B}{}^{G}\phi_{CD)FG} +
\phi_{(ABC}{}^{F}\bar{\xi}_{D)F}\right) &=\tfrac{1}{2}\phi\sigma
=0,\\ \iota^A\iota^B \iota^C \iota^D
\left(2\bar{\kappa}_{(A}{}^{F}\mathcal{D}_{B}{}^{G}\phi_{CD)FG} +
\phi_{(ABC}{}^{F}\bar{\xi}_{D)F}\right) &=\phi(\beta - \kappa + \tau)
= 0,
\end{align*}
with the second equality following from \eqref{NPRelationsTypeN} and
with all other components vanishing trivially. These are essentially
the same computations as in \cite{BaeVal10c}. On the other hand,
substituting \eqref{AuxSpinorsTypeN}, condition
\eqref{RedundantCondition2} reduces to
\begin{align*}
\omicron^A\omicron^B\iota^C\iota^D\left(\bar{\xi}^{FG}\mathcal{D}_{FG}\phi_{ABCD}
+
\tfrac{2}{3}\bar{\xi}\mathcal{D}_{(A}{}^F\phi_{BCD)F}\right)&=\tfrac{1}{3}\phi
\xi
\sigma=0,\\ \omicron^A\iota^B\iota^C\iota^D\left(\bar{\xi}^{FG}\mathcal{D}_{FG}\phi_{ABCD}
+ \tfrac{2}{3}\bar{\xi}\mathcal{D}_{(A}{}^F\phi_{BCD)F}\right)&=
\tfrac{1}{2} \beta \mathcal{D}_{\bm0\bm0}\phi-2 \beta^2 \phi -
\tfrac{1}{2} \beta \kappa \phi - 2 \rho \sigma \phi + \tfrac{1}{2}
\beta \tau \phi
=0,\\ \iota^A\iota^B\iota^C\iota^D\left(\bar{\xi}^{FG}\mathcal{D}_{FG}\phi_{ABCD}
+ \tfrac{2}{3}\bar{\xi}\mathcal{D}_{(A}{}^F\phi_{BCD)F}\right) &= 2
\rho \mathcal{D}_{\bm0\bm0}\phi-10 \beta \rho \phi =0,
\end{align*}
where we are again using \eqref{NPRelationsTypeN} and
\eqref{TypeNweylConstraint}. All other components vanish
trivially. Hence, in this case, both conditions
\eqref{RedundantCondition1} and \eqref{RedundantCondition2}
trivialise.

\subsubsection{Case II: $\bar{\kappa}_{AB}\bar{\kappa}^{AB}\neq 0$ on $\mathcal{U}$}\label{TypeDCase}

There exists a spin dyad $\lbrace \bm\omicron, \bm\iota\rbrace$ such
that $\bar{\kappa}_{AB} = e^{\varkappa} \omicron_{(A}\iota_{B)}$ for
some $\varkappa:\mathcal{U}\rightarrow\mathbb{C}$. The Buchdahl
constraint then implies that the rescaled Weyl spinor takes the form
\[ \phi_{ABCD}= \phi \omicron_{(A}\omicron_B\iota_C\iota_{D)}\]
for some scalar field $\phi:\mathcal{U}\rightarrow\mathbb{C}$ and
hence that the curvature is of Petrov type D on $\mathcal{U}$. The
equation $\mathcal{D}_{(AB}\bar{\kappa}_{CD)}=0$ is equivalent to
\begin{subequations}
\begin{eqnarray}
&& \sigma=0,\label{SpatialSenTypeD0000}\\
&& \mathcal{D}_{\bm0\bm0}\varkappa = \tau - \kappa, \label{SpatialSenTypeD0001}\\
&& \mathcal{D}_{\bm0\bm1}\varkappa = -\tfrac{1}{2} (\rho+\mu),\label{SpatialSenTypeD0011}\\
&& \mathcal{D}_{\bm1\bm1}\varkappa = \pi - \nu, \label{SpatialSenTypeD0111}\\
&& \lambda = 0.\label{SpatialSenTypeD1111}
\end{eqnarray}
\end{subequations}
Using the above, the auxiliary spinors can be written as 
\begin{subequations}
\begin{eqnarray} 
&& \bar{\xi} =-\tfrac{3}{2}e^{\varkappa}(\rho+\mu), \label{AuxSpinor1TypeD}\\
&& \bar{\xi}_{AB} = e^{\varkappa} (\nu - \pi) o_{A} o_{B} + e^{\varkappa}(\rho-\mu) o_{(A}\iota_{B)} + e^{\varkappa} (\tau - \kappa) \iota_{A} \iota_{B}. \label{AuxSpinor2TypeD}
\end{eqnarray}
\end{subequations}
The constraint $\mathcal{D}^{CD}\phi_{CDAB}=0$ is equivalent to 
\begin{equation}\label{WeylConstraintTypeD}
  \mathcal{D}_{\bm0\bm0}\phi = 3 \phi(\kappa - \tau), \qquad 
  \mathcal{D}_{\bm0\bm1}\phi = \tfrac{3}{2}\phi(\rho + \mu), \qquad
  \mathcal{D}_{\bm1\bm1}\phi = 3\phi (\nu  - \pi). 
\end{equation}
Then, substituting \eqref{AuxSpinor1TypeD} and \eqref{AuxSpinor2TypeD}, condition \eqref{RedundantCondition1} decomposes as follows
\begin{align*}
\omicron^A \omicron^B \omicron^C \omicron^D \left(2\bar{\kappa}_{(A}{}^{F}\mathcal{D}_{B}{}^{G}\phi_{CD)FG} + \phi_{(ABC}{}^{F}\bar{\xi}_{D)F}\right) 
&=- \tfrac{1}{2} e^{\varkappa}\phi \sigma=0,\\
\omicron^A \omicron^B \omicron^C \iota^D \left(2\bar{\kappa}_{(A}{}^{F}\mathcal{D}_{B}{}^{G}\phi_{CD)FG} + \phi_{(ABC}{}^{F}\bar{\xi}_{D)F}\right) & =\tfrac{1}{24} e^{\varkappa} \mathcal{D}_{\bm0\bm0}\phi + \tfrac{1}{8} e^{\varkappa}\phi (\tau - \kappa)  =0,\\
 \omicron^A \omicron^B \iota^C \iota^D \left(2\bar{\kappa}_{(A}{}^{F}\mathcal{D}_{B}{}^{G}\phi_{CD)FG} + \phi_{(ABC}{}^{F}\bar{\xi}_{D)F}\right) 
& = \tfrac{1}{18} e^{\varkappa} \mathcal{D}_{\bm0\bm1}\phi - \tfrac{1}{12} e^{\varkappa}\phi (\rho + \mu)  = 0,\\
 \omicron^A \iota^B \iota^C \iota^D \left(2\bar{\kappa}_{(A}{}^{F}\mathcal{D}_{B}{}^{G}\phi_{CD)FG} + \phi_{(ABC}{}^{F}\bar{\xi}_{D)F}\right) & =\tfrac{1}{24} e^{\varkappa} \mathcal{D}_{\bm1\bm1}\phi + \tfrac{1}{8} e^{\varkappa} \phi(\pi-\nu)  = 0,\\
 \iota^A \iota^B \iota^C \iota^D \left(2\bar{\kappa}_{(A}{}^{F}\mathcal{D}_{B}{}^{G}\phi_{CD)FG} + \phi_{(ABC}{}^{F}\bar{\xi}_{D)F}\right) & = - \tfrac{1}{2} e^{\kappa} \phi \lambda = 0,
\end{align*}
the equality with zero following from
\eqref{WeylConstraintTypeD}. Again, these are essentially the same
computations as in \cite{BaeVal10c}.  On the other hand, substituting
\eqref{AuxSpinor1TypeD} and \eqref{AuxSpinor2TypeD} into
\eqref{RedundantCondition2},
\begin{align*}
& o^Ao^Bo^Co^D(\bar{\xi}^{FG}\mathcal{D}_{FG}\phi_{ABCD} + \tfrac{2}{3} \bar{\xi} \mathcal{D}_{(A}{}^{F}\phi_{BCD)F})  =\tfrac{1}{3} \phi \sigma  \bar{\xi} =0,\\
& o^Ao^Bo^C\iota^D(\bar{\xi}^{FG}\mathcal{D}_{FG}\phi_{ABCD} + \tfrac{2}{3} \bar{\xi} \mathcal{D}_{(A}{}^{F}\phi_{BCD)F}) \\
&\qquad\qquad\qquad\qquad\qquad = \tfrac{1}{8} e^{\varkappa} (\rho+\mu) \mathcal{D}_{\bm0\bm0}\phi + \tfrac{3}{8}e^{\varkappa}\phi  (\tau - \kappa)(\rho+\mu)  +  \tfrac{1}{2} e^{\varkappa}\phi(\pi-\nu) \sigma   = 0	,\\
& o^Ao^B\iota^C\iota^D(\bar{\xi}^{FG}\mathcal{D}_{FG}\phi_{ABCD} + \tfrac{2}{3} \bar{\xi} \mathcal{D}_{(A}{}^{F}\phi_{BCD)F})  \\
&\qquad\qquad\qquad\qquad\qquad=	\tfrac{1}{6} e^{\varkappa}\left((\nu-\pi)\mathcal{D}_{\bm0\bm0}\phi + (\rho-\mu)\mathcal{D}_{\bm0\bm1}\phi + (\tau-\kappa)\mathcal{D}_{\bm1\bm1}\phi \right)  + \tfrac{1}{4} e^{\varkappa} \phi(\mu^2-\rho^2) = 0,\\
& o^A\iota^B\iota^C\iota^D(\bar{\xi}^{FG}\mathcal{D}_{FG}\phi_{ABCD} + \tfrac{2}{3} \bar{\xi} \mathcal{D}_{(A}{}^{F}\phi_{BCD)F})  \\
& \qquad\qquad\qquad\qquad\qquad =	- \tfrac{1}{8} e^{\varkappa}  (\rho+\mu) \mathcal{D}_{\bm1\bm1}\phi  +  \tfrac{3}{8} e^{\varkappa} \phi(\rho+\mu) (\nu-\pi) + \tfrac{1}{2} e^{\varkappa}\phi (\kappa-\tau) \lambda    = 0,\\
& \iota^A\iota^B\iota^C\iota^D(\bar{\xi}^{FG}\mathcal{D}_{FG}\phi_{ABCD} + \tfrac{2}{3} \bar{\xi} \mathcal{D}_{(A}{}^{F}\phi_{BCD)F})  =	- \tfrac{1}{3}\phi \lambda \bar{\xi} = 0.
\end{align*}
Hence, again, \eqref{RedundantCondition1} and
\eqref{RedundantCondition2} trivialise. Combining the result of this
section with the previous, Lemma
\ref{prop_remove_redundant_conditions} follows immediately.


\end{document}